\pdfoutput=1
\RequirePackage{ifpdf}
\ifpdf 
\documentclass[pdftex]{sigma}
\else
\documentclass{sigma}
\fi

\numberwithin{equation}{section}

\newtheorem{Theorem}{Theorem}[section]
\newtheorem{Corollary}[Theorem]{Corollary}
\newtheorem{Proposition}[Theorem]{Proposition}
 { \theoremstyle{definition}
\newtheorem{Remark}[Theorem]{Remark} }

\begin{document}


\renewcommand{\thefootnote}{$\star$}

\newcommand{\arXivNumber}{1511.08098}

\renewcommand{\PaperNumber}{054}

\FirstPageHeading

\ShortArticleName{Multidimensional Toda Lattices: Continuous and Discrete Time}

\ArticleName{Multidimensional Toda Lattices: \\ Continuous and Discrete Time\footnote{This paper is a~contribution to the Special Issue on Orthogonal Polynomials, Special Functions and Applications.
The full collection is available at \href{http://www.emis.de/journals/SIGMA/OPSFA2015.html}{http://www.emis.de/journals/SIGMA/OPSFA2015.html}}}

\Author{Alexander I.~APTEKAREV~$^{\dag^1}$, Maxim DEREVYAGIN~$^{\dag^2}$, Hiroshi MIKI~$^{\dag^3}$\\
 and Walter VAN ASSCHE~$^{\dag^4}$}

\AuthorNameForHeading{A.I.~Aptekarev, M.~Derevyagin, H.~Miki and W.~Van Assche}

\Address{$^{\dag^1}$~Keldysh Institute for Applied Mathematics, Russian Academy of Sciences,\\
\hphantom{$^{\dag^1}$}~Miusskaya pl.~4, 125047 Moscow, Russia}
\EmailDD{\href{mailto:aptekaa@keldysh.ru}{aptekaa@keldysh.ru}}

\Address{$^{\dag^2}$~University of Mississippi, Department of Mathematics,\\
\hphantom{$^{\dag^2}$}~Hume Hall 305, P. O. Box 1848, University, MS 38677-1848, USA}
\EmailDD{\href{mailto:derevyagin.m@gmail.com}{derevyagin.m@gmail.com}}

\Address{$^{\dag^3}$~Doshisha University, Department of Electronics, Faculty of Science and Engineering,\\
\hphantom{$^{\dag^3}$}~Kyotanabe city, Kyoto 610 0394, Japan}
\EmailDD{\href{mailto:hmiki@mail.doshisha.ac.jp}{hmiki@mail.doshisha.ac.jp}}

\Address{$^{\dag^4}$~KU Leuven, Department of Mathematics, Celestijnenlaan 200B box 2400,\\
\hphantom{$^{\dag^4}$}~BE-3001 Leuven, Belgium}
\EmailDD{\href{mailto:walter.vanassche@kuleuven.be}{walter.vanassche@kuleuven.be}}

\ArticleDates{Received January 05, 2016, in f\/inal form June 01, 2016; Published online June 13, 2016}

\Abstract{In this paper we present multidimensional analogues of both the continuous- and discrete-time Toda lattices. The integrable systems that we consider here have two or more space coordinates. To construct the systems, we generalize the orthogonal polynomial approach for the continuous and discrete Toda lattices to the case of multiple orthogonal polynomials.}

\Keywords{multiple orthogonal polynomials; orthogonal polynomials; recurrence relations; Toda equation; discrete integrable system; Toda lattice}

\Classification{42C05; 37K10; 39A14; 65Q10}

\renewcommand{\thefootnote}{\arabic{footnote}}
\setcounter{footnote}{0}

\section{Introduction}

The continuous-time Toda lattice \cite{Toda,Toda2}
\begin{gather}
 \dot{a}_n(t)=a_n(t)(b_{n-1}(t)-b_{n}(t)), \nonumber\\
 \dot{b}_n(t)=a_{n}(t)-a_{n-1}(t),\qquad a_n>0, \qquad t\in \mathbb{R}_+, \qquad n\in{\mathbb Z}_+, \label{toda}
\end{gather}
and the discrete-time Toda lattice \cite{Hirota,Spiridonov2}
\begin{gather}
A_{n}^{t+1}+B_n^{t+1}=A_{n}^t+B_{n+1}^t, \nonumber \\
A_{n-1}^{t+1}B_n^{t+1}=A_n^tB_n^t,\qquad t, n\in{\mathbb Z}_+, \label{dtoda}
\end{gather}
have \looseness=-1 appeared in physical and mathematical models quite a while ago and are still attracting the interest of many researchers in the f\/ield. For instance, generalizations of the Toda lattice have been considered from various points of view (e.g.,~\cite{Jack,Santini}). One of the interesting aspects of the continuous- and discrete-time Toda lattice is that \textit{orthogonal polynomials} (OPs) appear as eigenfunctions of their Lax pairs~\cite{Aptekarev,Spiridonov2}, which means that spectral transformations of OPs describe the f\/low of the Toda lattice. In addition, several integrable systems have been shown to be related to descendants of OPs through their spectral transformations~\cite{Aptekarev2,RTL1,RTL2,RTL3}. Applying this spectral transformation technique to a new class of OPs, novel integrable systems have been exploited~\cite{Adler,Spiridonov2,SpZh}. Recall (see~\cite{kaka, moser}) that the system~\eqref{toda} is managed by the evolution
\begin{gather}\label{evol}
d\mu(x,t)=e^{-xt} d\mu(x) ,
\end{gather}
and the functions $a_n(t)$ and $b_n(t)$ appear as the coef\/f\/icients of the \textit{three-term recurrence relation}
\begin{gather}\label{recOP}
xP_{n}(x;t)=P_{n+1}(x;t)+b_{{n}}(t)P_{{n}}(x;t)+a_{{n}}(t)P_{{n}-1}(x;t), \qquad n\in{\mathbb Z}_+,
\end{gather}
where $P_{{n}}(x;t)$ are monic polynomials in the variable $x$, orthogonal with respect to $d\mu(x,t)$. Thus, the direct and inverse spectral transformations $\{a_{{n}}(t), b_{{n}}(t)\} \rightleftarrows d\mu(t,x)$ along with the evolution~\eqref{evol} solve the Cauchy problem for~\eqref{toda}.

In this paper, motivated by these results, we aim to explore a new generalization of the Toda lattice by developing the spectral transformations of \textit{multiple orthogonal polynomials} \mbox{(m-OPs)}~\mbox{\cite{6, MI}}, which originated from the theory of \textit{Hermite--Pad\'{e} approximants}~\cite{Nikishin}, which were introduced by Hermite~\cite{1} in connection with his proof of the transcendence of~$e$. On top of that, m-OPs were recently found to have applications in many areas such as random matri\-ces~\cite{ABK, Bleher} and the theory of dif\/ference operators on lattices~\cite{Aptekar, ADvanA2015}.

A lattice of multiple orthogonal polynomials $P_{\vec{n}}$, $\vec{n}\in {\mathbb Z}_+^r$ is def\/ined as a set of polynomials of degree $|\vec{n}|=n_1+\cdots +n_r$, satisfying the orthogonality relations
\begin{gather}\label{orthogonality}
\mathcal{L}_j[x^iP_{\vec{n}}(x)]=\int x^iP_{\vec{n}}(x) d\mu_j(x)=0,\qquad i=0,\dots ,n_j-1,\qquad 1\leq j\leq r.
\end{gather}
This def\/inition gives $|\vec{n}|$ homogeneous and linear equations for the $|\vec{n}|+1$ coef\/f\/icients of the polynomial $P_{\vec{n}}$.
In fact, we can choose $P_{\vec{n}}$ to be monic and, so, it leads to a linear system of~$|\vec{n}|$ unknown coef\/f\/icients, which has a unique solution if and only if the corresponding determinant is not vanishing. Evidently, it is not always the case that the determinant is nonzero. Therefore, there is no guarantee that for a given multi-index one can f\/ind the corresponding multiple orthogonal polynomial. In the case of uniqueness the multi-index~$\vec{n}$ is called \textit{normal}. Hence, for normal indices the monic polynomial $P_{\vec{n}}$ can be determined. If all multi-indices $\vec{n}\in {\mathbb Z}_+^r$ on the lattice are normal then the system of measures $\{\mu_j\}_{j=1}^r$ (or functionals $\{\mathcal{L}_j\}_{j=1}^r$) generating the lattice of polynomials $\{P_{\vec{n}}\}$ in~\eqref{orthogonality} is called a~\textit{perfect system}. In other words, for a perfect system the polynomial $P_{\vec{n}}$ is well def\/ined for any $\vec{n}\in {\mathbb Z}_+^r$. This notion has been introduced by K.~Mahler~\cite{2} in relation to the diophantine approximation.

It is noteworthy that m-OPs become ordinary OPs if we take $r=1$. Moreover, one of the properties inherited from ordinary OPs is that m-OPs satisfy the following \textit{nearest-neighbor recurrence relations} \cite{MI, vanA2011}, which actually generalizes the three-term recurrence relation for OPs~\eqref{recOP}:
\begin{gather}\label{rec}
xP_{\vec{n}}(x)=P_{\vec{n}+\vec{e}_i}(x)+b_{\vec{n},i}P_{\vec{n}}(x)+
\sum_{k=1}^ra_{\vec{n},k}P_{\vec{n}-\vec{e}_k}(x), \qquad i=1,\dots ,r,\qquad \vec{n}\in{\mathbb Z}_+^r,
\end{gather}
where $\vec{e}_i=(0,\dots, 0,1,0,\dots ,0)$ is the $i$-th vector of the standard basis in~${\mathbb Z}_+^r$. Unlike in the ordinary case, the coef\/f\/icients of the recurrence relations $\{ a_{\vec{n},i},b_{\vec{n},i}\}$ for m-OPs are not independent but are required to satisfy the dif\/ference equations~\cite{vanA2011}
\begin{gather}
 \sum_{k=1}^r (a_{\vec{n}+\vec{e}_j,k}-a_{\vec{n}+\vec{e}_i,k})=
(b_{\vec{n},j}-b_{\vec{n},i}) (b_{\vec{n}+\vec{e}_j,i}-b_{\vec{n},i}),\nonumber\\
\frac{a_{\vec{n},i}}{a_{\vec{n}+\vec{e}_j,i}}=
\frac{b_{\vec{n}-\vec{e}_i,j}-b_{\vec{n}-\vec{e}_i,i}}{b_{\vec{n},j}-b_{\vec{n},i}},
\qquad 1 \leq i\neq j \leq r.\label{diffcoeff}
\end{gather}
Observe that \eqref{diffcoeff} is equivalent to the equations in~\cite[Theorem~3.2]{vanA2011}: interchanging $i$ and $j$ in the f\/irst equation of~\eqref{diffcoeff} gives~(3.6) from~\cite{vanA2011} and the determinant in~(3.7) from~\cite{vanA2011} is the same as the right hand side of the f\/irst equation in our~\eqref{diffcoeff}.

Finally, we are in a position to state the main results of the paper.

\begin{Theorem}\label{T1}
Suppose that the system of measures
\begin{gather}\label{I_2.1}
d\mu_j(x,t)=e^{-tx}d\mu_j(x,0),\qquad 1\leq j\leq r,
\end{gather}
generates m-OPs with normal indices $\vec{n}$ and $ \{\vec{n}\pm\vec{e}_k\}_{k=1}^r$ in a neighborhood of $t=0$. Then the recurrence coefficients from~\eqref{rec} satisfy locally the equations
\begin{gather}
\dot{a}_{\vec{n},k}=a_{\vec{n},k}[b_{\vec{n}-\vec{e}_k,k}-b_{\vec{n},k}], \nonumber\\
\dot{b}_{\vec{n},k}=\sum_{j=1}^r (a_{\vec{n},j}-a_{\vec{n}+\vec{e}_k,j}), \qquad 1\leq k\leq r.\label{I_2.2}
\end{gather}
If the system $\{d\mu_j(x,t)\}_{j=1}^{r}$ in~\eqref{I_2.1} is
perfect for $t\geqslant 0$, then the Cauchy problem for~\eqref{I_2.2} has a global solution which can be obtained by the direct and inverse spectral transformations $\{a_{\vec{n},k}(t), b_{\vec{n},k}(t)\}\rightleftarrows\{d\mu_k(x,t)\}$ and the evolution~\eqref{I_2.1}.
\end{Theorem}

\begin{Remark}\label{RemT1-1}
We would like to emphasize here that the system~(\ref{I_2.2}) is solvable as long as we choose the initial values subject to (\ref{diffcoeff}), which in turn means that we have to start with a~perfect system of measures. Once again, such systems generate coef\/f\/icients that satisfy \eqref{diffcoeff}. Moreover, for the global solution the Toda dynamics preserve the stationary equations~(\ref{diffcoeff}). Therefore,~(\ref{I_2.2}) forms a ``weakly integrable'' system.
Similar phenomena occur for the quantum Hamiltonians associated with classical multiple orthogonal polynomials~\cite{Miki2,Miki,FranWVA}. Also, it should be noted that there are other multidimensional integrable systems studied in the literature. For instance, see~\cite{Manakov}; see also~\cite{Apt2016} and references therein.
\end{Remark}

\begin{Remark}\label{RemT1-2}
There are two well-known perfect systems of measures. One of them is called an Angelesco system~\cite{21} and is formed by measures with supports on disjoint intervals. The perfectness of an Angelesco system immediately follows from~\eqref{orthogonality} and the properties of zeros of OPs. The other one is called a Nikishin system~\cite{23}. The measures from a Nikishin system have the same support but some extra conditions (analytic properties of the weight functions) need to hold. Details of the def\/inition of Nikishin systems and the proof of their perfectness can be found in~\cite{FL}.
\end{Remark}

\begin{Remark}\label{RemT1-3}
The system $\{d\mu_j\}_{j=1}^{r}$ consists of the \textit{spectral measures} of the $r$ marginal one-dimensional dif\/ference operators def\/ined by the three-term recurrence relations~\eqref{recOP} with coef\/f\/i\-cients $a_{n}^{(j)}:=a_{ne_{j},j}$, $b_{n}^{(j)}:=b_{ne_{j},j}$, $j=1,\dots, r$. These marginal spectral measures can be taken as \textit{spectral data} for the multidimensional dif\/ference operator def\/ined by the recurrence rela\-tions~\eqref{rec} with
the coef\/f\/icients $\{a_{\vec{n},k}, b_{\vec{n},k}\}$. Therefore, the direct and inverse spectral transformations $\{a_{\vec{n},k}(t), b_{\vec{n},k}(t)\}\rightleftarrows\{d\mu_k(t,x)\}$ reduce to the well-known direct and inverse spectral problems for OPs together with a scheme to solve a boundary value problem (BVP) for the discrete integrable system~\eqref{diffcoeff}; see \cite{ADvanA2014, FHVanA} for more details on this matter.
\end{Remark}

The proof of Theorem~\ref{T1} and properties of m-OPs when the corresponding measures evolve as in \eqref{I_2.1} are presented in Section~\ref{sec:2} (see Subsections~\ref{subsec:2.1}--\ref{subsec:2.3}).

It turns out that an important place on the lattice of m-OPs \eqref{orthogonality}--\eqref{rec} is the diagonal. Namely, the \textit{diagonal sequence} $\{q_N\}$ of multiple orthogonal polynomials~\cite{3} (also called \textit{$r$-orthogonal polynomials}~\cite{Cheikh-I, Cheikh-II,Douak}) that is generated in the following manner
\begin{gather}\label{dop}
q_{N} =P_{(n,\dots ,n)+\sum\limits_{j=1}^k e_j},\qquad N =: nr+k, \qquad k=0,\dots,r-1.
\end{gather}
The diagonal m-OPs appear as the common denominator of the convergents of the underlying Jacobi--Perron vector continued fraction \cite{Bern, Nikishin}, which is a functional analog of the continued fraction introduced by Jacobi \cite{Jac1, Jac2} and Perron~\cite{Perr} in their approach to f\/ind a characterization of irrationals of orders higher than quadratic. As a matter of fact, the polynomial sequence $\{q_N\}$ satisf\/ies the \textit{step-line recurrence relation}~\cite{6, Nikishin}
\begin{align}\label{recStepLine}
xq_{N}(x)=q_{N+1}(x)+\beta_{N}q_{N}(x)+
\sum_{k=1}^r \alpha_{N}^{(k)}q_{N-k}(x),\qquad n\in{\mathbb Z}_+.
\end{align}
If the system of measures $\{d\mu_j\}_{j=1}^{r}$ evolves subject to \eqref{I_2.1}, then (see \cite{Adler2}) the equations for the coef\/f\/icients from~\eqref{recStepLine} can be written in the Lax pair form
\begin{gather}\label{LAtoda}
\dot{L}=[L,(L)_{-}], \qquad L := L\big(1, \beta_{N}, \alpha_{N}^{(1)},\dots, \alpha_{N}^{(r)}\big),
\end{gather}
where $(X)_-$ denotes the strictly lower part of the matrix $X$ and $L$ is a $(r+2)$-banded (lower Hessenberg) semi-inf\/inite matrix with 1's on the upper diagonal, $\{\beta_{N}\}$ on the main diagonal and $\{\alpha_{N}^{(k)}\}_{k=1}^r$ building the lower diagonals. In Subsection~\ref{subsec:2.4} we show relations between two generalizations of the Toda equations, which are basically the relations between~\eqref{I_2.2} and~\eqref{LAtoda}.

In the same spirit as it is done for the continuous-time Toda equation, we obtain in Section~\ref{sec:3} a discrete analogue of \eqref{I_2.1}
and~\eqref{I_2.2}.
\begin{Theorem}\label{T2}
Suppose that the system of measures
\begin{gather*}
d\mu_j(x,t)=x^{t}d\mu_j(x,0),\qquad t\in{\mathbb Z}_+, \qquad 1\leq j\leq r,
\end{gather*}
generates the m-OPs $\{P_{\vec{n}}^t(x)\}$ with normal indices $\vec{n}$ and $ \{\vec{n}\pm\vec{e}_k\}_{k=1}^r$. Then the following nonlinear dif\/ference system on a semi-infinite lattice
\begin{gather}
A_{\vec{n},j}^{t+1}+\sum_{k=1}^rB_{\vec{n},k}^{t+1}=A_{\vec{n},j}^t+\sum_{k=1}^rB_{\vec{n}+\vec{e_j},k}^t,\nonumber\\
A_{\vec{n}-\vec{e}_j,j}^{t+1}B_{\vec{n},j}^{t+1}=A_{\vec{n},j}^tB_{\vec{n},j}^t,
\qquad 1 \leq j \leq r,\label{Int_DToda}
\end{gather}
 can be solved locally with respect to the space variable. Moreover, if the system of measures is perfect for any discrete time $t\in{\mathbb Z}_+$, then the solution exists globally and it is given by the formulas
\begin{gather*}
A_{\vec{n},j}^t=-\frac{P_{\vec{n}+\vec{e}_j}^t(0)}{ P_{\vec{n}}^t(0) },\qquad
B_{\vec{n},k}^t=-a_{\vec{n},k}^t\frac{P_{\vec{n}-\vec{e}_k}^t(0)}{P_{\vec{n}}^t(0)}.
\end{gather*}
\end{Theorem}

To get to \eqref{Int_DToda} we present two approaches. The f\/irst one is based on Christof\/fel and Geronimus transformations, which are known to be discrete analogues of Darboux transformations, see Subsection~\ref{subsec:3.2}. The second method is obtained by following the consistency approach from~\cite{BS2002} and~\cite{SNderK2011}. In particular, the Lax pair we get for the second method is a certain adaptation of the one from~\cite{SNderK2011} to our setting.
The latter approach allows us to present the discrete-time Toda equations~\eqref{Int_DToda} and the consistency equations~\eqref{diffcoeff} in a unif\/ied fashion in the form of Lax pairs commutation relations, see Subsection~\ref{subsec:3.4}. It is worth mentioning that both the approaches are related to a generalization of the quotient-dif\/ference algorithm for the Pad\'{e} table. Moreover, our approach complements some earlier attempts to develop the q-d algorithm~\cite{I2003} related to multiple orthogonal polynomials and puts it into the context of discrete integrable systems.

Also, in Subsection~\ref{subsec:3.3} we connect the discrete-time multidimensional Toda equations \eqref{Int_DToda} and the discrete time analogue of Toda chain equation~\eqref{LAtoda} for the recurrence coef\/f\/icients of the diagonal sequence of m-OPs \eqref{recStepLine}. Using the formulas for multiple Laguerre polyno\-mials~\cite{ABVA, vanA2011} we give the initial data for the explicit solution of the multidimensional Toda lattice in both cases. For the continuous time it can be found in Subsection~\ref{subsec:2.5} and for the discrete time in Subsection~\ref{subsec:3.5}.

\section{The continuous-time higher analogues of the Toda lattice}\label{sec:2}

\subsection{Preliminaries}\label{subsec:2.1}
Let us brief\/ly go over the basic def\/initions, see \cite[Chapter 23]{MI} for details. The m-OPs $P_{\vec{n}}$ def\/ined in \eqref{orthogonality} are called multiple orthogonal (or Hermite--Pad\'{e}) polynomials of \textit{type II}. In what follows, we suppose that~$P_{\vec{n}}$ are monic so that the m-OPs are uniquely determined. It is useful to consider the dual construction. More precisely, \textit{type~I} multiple orthogonal polynomials $(C_{\vec{n}, 1}, \dots,C_{\vec{n},r})$ are such that~$C_{\vec{n},j}$ is a polynomial of degree at most ${n}_j - 1$ with the orthogonality relations
\begin{gather*}
\sum^r_{j=1}\int x^kC_{\vec{n},j}(x) d\mu_j (x) = 0,\qquad k= 0, 1, \dots , |\vec{n}| - 2,
\end{gather*}
and the normalization
\begin{gather*}
\sum^r_{j=1}\int x^{ |\vec{n}| -1}C_{\vec{n},j}(x) d\mu_j (x) = 1, \qquad
C_{\vec{n},j}(x)=\kappa_{\vec{n},j}x^{n_j-1}+\cdots .
\end{gather*}
We assume that the $r$ measures $\mu_1, \dots,\mu_r$ are all absolutely continuous with respect to a~measu\-re~$\mu$ and that $d\mu_j (x) = w_j(x) d\mu (x) $ and we will use the following notation
\begin{gather*}Q_{\vec{n}}(x) = \sum^r_{j=1}C_{\vec{n},j}(x)w_j(x).
\end{gather*}

It is easy to check that the type I and type II multiple orthogonal polynomials form a~ \textit{bi\-ortho\-nor\-mal system} in the sense that the functions $Q_{\vec{n}}$ generated by type I multiple orthogonal polynomials satisfy
\begin{gather}\label{biorth}
\int P_{\vec{n}}(x) Q_{\vec{m}}(x) d\mu (x) = \begin{cases}
0& \text {if } |\vec{m}|\leq |\vec{n}|,\\
0& \text{if } |\vec{n}|\leq |\vec{m}| - 2,\\
1 &\text{if } \vec{m} = \vec{n} +\vec{e}_k \text{ for } 1\leq k\leq r.
 \end{cases}
\end{gather}
Furthermore, for the type II polynomials we have \eqref{rec}:
\begin{gather}\label{1.1}
xP_{\vec{n}}(x) =P_{ \vec{n} +\vec{e}_k}(x)+b_{\vec{n},k}P_{\vec{n}}(x)+ \sum^r_{j=1}a_{\vec{n},j}P_{ \vec{n} -\vec{e}_j}(x),
\end{gather}
and using the type I polynomials we get
\begin{gather}\label{1.2}
b_{\vec{n},k}=\int xP_{\vec{n}} Q_{\vec{n}+\vec{e}_k}(x) d\mu (x),
\qquad
a_{\vec{n},j}=\frac{\int x^{ n_j }P_{\vec{n}}(x) d\mu_j (x)}{\int x^{ n_j -1}P_{\vec{n}-\vec{e}_k}(x) d\mu_j (x)}.
\end{gather}
For type I we have a similar recurrence relation
\begin{gather}\label{1.3}
xQ_{\vec{n}}(x) =Q_{ \vec{n} -\vec{e}_k}(x)+b_{\vec{n}-\vec{e}_k,k}Q_{\vec{n}}(x)+ \sum^r_{j=1}a_{\vec{n},j}Q_{ \vec{n} +\vec{e}_j}(x).
\end{gather}
We need to point out two useful relations here. If we multiply \eqref{1.1} by $Q_{\vec{n}}(x)$ and integrate, then the biorthogonality gives
\begin{gather}\label{1.5}
\int xP_{\vec{n}}(x) Q_{\vec{n}}(x) d\mu (x)=\sum_{j=1}^ra_{\vec{n},j}.
\end{gather}
The orthogonality properties of $P_{\vec{n}}$ imply that
\begin{gather*}
\int P_{\vec{n}}(x) Q_{\vec{n}+\vec{e}_k}(x) d\mu (x) =\int P_{\vec{n}} C_{\vec{n}+\vec{e}_k,k}(x) d\mu_k (x),
\end{gather*}
so we arrive at
\begin{gather}\label{1.6}
1=\kappa_{\vec{n}+\vec{e}_k,k}\int P_{\vec{n}}(x) x^{ n_k} d\mu_k (x).
\end{gather}

\subsection{Time dynamics}\label{subsec:2.2}
In this subsection we present a \textit{proof of Theorem}~\ref{T1}. To this end, we consider a perfect system of measures $(\mu_1(x,t), \dots,\mu_r(x,t))$ that depend on time as in \eqref{I_2.1}, i.e., $d\mu_k(x,t):=e^{-tx} d\mu_k(x)$ for $1\leq k\leq r$. Now the corresponding type I and type II
multiple orthogonal polynomials depend on time as well, that is, we have $Q_{\vec{n}}(x;t)$ and $P_{\vec{n}}(x;t)$. Our goal is to show that the relations \eqref{1.2} evaluated at the moment~$t$
\begin{gather*}
b_{\vec{n},k}(t)=\int xP_{\vec{n}} (x;t) Q_{\vec{n}+\vec{e}_k}(x;t)e^{-xt} d\mu (x),
\qquad
a_{\vec{n},k}(t)=\frac{\int x^{ n_k}P_{\vec{n}}(x;t) e^{-xt} d\mu_k (x)}{\int x^{ n_k -1}P_{\vec{n}-\vec{e}_k}(x) e^{-xt} d\mu_k (x)},
\end{gather*}
imply\textit{ multiple Toda} (m-Toda) equations \eqref{I_2.2}, i.e., for $1\leq k\leq r$
\begin{gather}\label{2.1}
\dot{a}_{\vec{n},k}=a_{\vec{n},k}[b_{\vec{n}-\vec{e}_k,k}-b_{\vec{n},k}],
\end{gather}
and
\begin{gather}\label{2.2}
\dot{b}_{\vec{n},k}=\sum_{j=1}^r (a_{\vec{n},j}-a_{\vec{n}+\vec{e}_k,j}).
\end{gather}

From the biorthogonality \eqref{biorth} we get
\begin{gather*}
\int P_{\vec{n}} (x;t) Q_{\vec{n}+\vec{e}_k}(x;t)e^{-xt} d\mu (x) =1.
\end{gather*}
If we take derivatives with respect to $t$, then this gives
\begin{gather*}
\int \dot {P}_{\vec{n}} (x;t) Q_{\vec{n}+\vec{e}_k}(x;t)e^{-xt} d\mu (x) +
\int P_{\vec{n}} (x;t) \dot{Q}_{\vec{n}+\vec{e}_k}(x;t)e^{-xt} d\mu (x)
 \\
\qquad{} =\int xP_{\vec{n}} (x;t) Q_{\vec{n}+\vec{e}_k}(x;t)e^{-xt} d\mu (x).
\end{gather*}
Observe that $\dot {P}_{\vec{n}}$ is a polynomial of degree at most $ |\vec{n}| - 1$ since ${P}_{\vec{n}}$ is a monic polynomial, hence the orthogonality for type I multiple orthogonal polynomials gives
\begin{gather*}
\int \dot {P}_{\vec{n}} (x;t) Q_{\vec{n}+\vec{e}_k}(x;t)e^{-xt} d\mu (x)=0.
\end{gather*}
The orthogonality for type II multiple orthogonal polynomials gives
\begin{gather}\label{bla}
\int P_{\vec{n}} (x;t) \dot{Q}_{\vec{n}+\vec{e}_k}(x;t)e^{-xt} d\mu (x) =\int P_{\vec{n}} (x;t) \dot{A}_{\vec{n}+\vec{e}_k,k}(x;t)e^{-xt} d\mu_k (x).
\end{gather}
If we use $C_{\vec{n}+\vec{e}_k,k}(x;t)=\kappa_{\vec{n}+\vec{e}_k,k}(t)x^{n_k}+\cdots$,
then we f\/ind
\begin{gather*}
\int P_{\vec{n}} (x;t) \dot{A}_{\vec{n}+\vec{e}_k,k}(x;t)e^{-xt} d\mu_k (x)=\kappa_{\vec{n}+\vec{e}_k,k}\int x^{ n_k}P_{\vec{n}}(x;t) e^{-xt} d\mu_k (x)=\frac{\dot{\kappa}_{\vec{n},j}}{\kappa_{\vec{n}+\vec{e}_k,j}},
\end{gather*}
where the last equality follows from~\eqref{1.5}. Combining these results with~\eqref{1.2} already gives
\begin{gather}\label{blabla}
b_{\vec{n},k}(t)=\frac{\dot{\kappa}_{\vec{n}+\vec{e}_k,k}}{\kappa_{\vec{n}+\vec{e}_k,k}}.
\end{gather}
Taking derivatives in \eqref{1.6} gives
\begin{gather*}
\frac{\dot{\kappa}_{\vec{n},k}}{\kappa_{\vec{n},k}}-\frac{\dot{\kappa}_{\vec{n}+\vec{e}_k,k}}{\kappa_{\vec{n}+\vec{e}_k,k}}
=\frac{\dot{a}_{\vec{n},k}}{a_{\vec{n},k}},
\end{gather*}
which implies
\begin{gather*}
b_{\vec{n}-\vec{e}_k,k}(t)-b_{\vec{n},k}(t)=\frac{\dot{a}_{\vec{n}+\vec{e}_k,k}}{a_{\vec{n}+\vec{e}_k,k}},
\end{gather*}
giving \eqref{2.1}. Thus the f\/irst relation in \eqref{I_2.2} is proved.

For the second relation in \eqref{I_2.2} we use \eqref{1.2} to f\/ind
\begin{gather*}
b_{\vec{n},k}(t)=\int xP_{\vec{n}} (x;t) Q_{\vec{n}+\vec{e}_k}(x;t)e^{-xt} d\mu (x).
\end{gather*}
Taking the derivative with respect to $t$ gives
\begin{gather*}
\dot{b}_{\vec{n},k}(t)=\int x\dot {P}_{\vec{n}} (x;t) Q_{\vec{n}+\vec{e}_k}(x;t)e^{-xt} d\mu (x)\\
\hphantom{\dot{b}_{\vec{n},k}(t)=}{} + \int xP_{\vec{n}} (x;t) \dot{Q}_{\vec{n}+\vec{e}_k}(x;t)e^{-xt} d\mu (x) -\int x^2 P_{\vec{n}} (x;t) Q_{\vec{n}+\vec{e}_k}(x;t)e^{-xt} d\mu (x).
\end{gather*}
If we use \eqref{1.3} then we f\/ind
\begin{gather*}
\int x\dot {P}_{\vec{n}} (x;t) Q_{\vec{n}+\vec{e}_k}(x;t)e^{-xt} d\mu (x)\\
\qquad{} = \int \dot {P}_{\vec{n}} (x;t)\left(Q_{ \vec{n}}(x;t)+b_{\vec{n},k}Q_{\vec{n}+\vec{e}_k}(x;t)+ \sum^r_{j=1}a_{\vec{n}+\vec{e}_k,j}Q_{ \vec{n} +\vec{e}_j+\vec{e}_k}(x;t)\right)e^{-xt} d\mu (x).
\end{gather*}
Since $\dot {P}_{\vec{n}}$ is of degree at most $ |\vec{n}| - 1$, the orthogonality for type~I multiple orthogonal
polynomials gives
\begin{gather*}
\int x\dot {P}_{\vec{n}} (x;t) Q_{\vec{n}+\vec{e}_k}(x;t)e^{-xt} d\mu (x) = \int \dot {P}_{\vec{n}} (x;t) Q_{\vec{n}}(x;t)e^{-xt} d\mu (x).
\end{gather*}
Taking the derivative of
\begin{gather*}
\int {P}_{\vec{n}} (x;t) Q_{\vec{n}}(x;t)e^{-xt} d\mu (x)=0
\end{gather*}
leads to
\begin{gather*}
\int \dot {P}_{\vec{n}} (x;t) Q_{\vec{n}}(x;t)e^{-xt} d\mu (x)=\int x {P}_{\vec{n}} (x;t) Q_{\vec{n}}(x;t)e^{-xt} d\mu (x)= \sum^r_{j=1}a_{\vec{n},j}(t),
\end{gather*}
hence
\begin{gather*}
\int x\dot {P}_{\vec{n}} (x;t) Q_{\vec{n}+\vec{e}_k}(x;t)e^{-xt} d\mu (x) = \sum^r_{j=1}a_{\vec{n},j}(t).
\end{gather*}
If we use \eqref{1.1} then we f\/ind
\begin{gather*}
\int x {P}_{\vec{n}} (x;t) \dot{Q}_{\vec{n}+\vec{e}_k}(x;t)e^{-xt} d\mu (x)\\
\qquad{} = \int\left(P_{ \vec{n}+\vec{e}_k}(x;t)+b_{\vec{n},k}(t)P_{\vec{n}}(x;t)+ \sum^r_{j=1}a_{\vec{n},j}P_{ \vec{n} -\vec{e}_j}(x;t)\right) \dot {Q}_{\vec{n}+\vec{e}_k} (x;t)e^{-xt} d\mu (x).
\end{gather*}
The orthogonality of type II multiple orthogonal polynomials gives
\begin{gather*}
 \int x {P}_{\vec{n}} (x;t) \dot{Q}_{\vec{n}+\vec{e}_k}(x;t)e^{-xt} d\mu (x)\\
\qquad{} =b_{\vec{n},k}\int {P}_{\vec{n}} (x;t) \dot{Q}_{\vec{n}+\vec{e}_k}(x;t)e^{-xt} d\mu (x)+ \sum^r_{j=1}a_{\vec{n},j}\int {P}_{\vec{n}-\vec{e}_j}(x;t)\dot{Q}_{ \vec{n} +\vec{e}_k}(x;t)e^{-xt} d\mu (x).
\end{gather*}
From \eqref{bla} and \eqref{blabla} we recall that
\begin{gather*}
\int {P}_{\vec{n}} (x;t) \dot{Q}_{\vec{n}+\vec{e}_k}(x;t)e^{-xt} d\mu (x) =b_{\vec{n},k}.
\end{gather*}
If we take the derivative of
\begin{gather*}
\int {P}_{\vec{n}-\vec{e}_j} (x;t) {Q}_{\vec{n}+\vec{e}_k}(x;t)e^{-xt} d\mu (x) =0,
\end{gather*}
then we f\/ind
\begin{gather*}
\int {P}_{\vec{n}-\vec{e}_j} (x;t) \dot{Q}_{\vec{n}+\vec{e}_k}(x;t)e^{-xt} d\mu (x)
=\int x {P}_{\vec{n}-\vec{e}_j} (x;t) {Q}_{\vec{n}+\vec{e}_k}(x;t)e^{-xt} d\mu (x),
\end{gather*}
and after using \eqref{1.1} and the biorthogonality, we f\/ind
\begin{gather*}
\int {P}_{\vec{n}-\vec{e}_j} (x;t) \dot{Q}_{\vec{n}+\vec{e}_k}(x;t)e^{-xt} d\mu (x) =1.
\end{gather*}
This gives
\begin{gather*}
\int x {P}_{\vec{n}} (x;t) \dot{Q}_{\vec{n}+\vec{e}_k}(x;t)e^{-xt} d\mu (x) =
{b}^2_{\vec{n},k}+\sum^r_{j=1}a_{\vec{n},j}.
\end{gather*}
Finally use \eqref{1.1}, \eqref{1.3} and the biorthogonality to f\/ind
\begin{gather*}
\int x^2 {P}_{\vec{n}} (x;t) {Q}_{\vec{n}+\vec{e}_k}(x;t)e^{-xt} d\mu (x) =
\sum^r_{j=1}a_{\vec{n}+\vec{e}_k,j}+{b}^2_{\vec{n},k}+\sum^r_{j=1}a_{\vec{n},j}.
\end{gather*}
Combining all these results then gives
\begin{gather*}
\dot{b}_{\vec{n},k}=\sum^r_{j=1}a_{\vec{n},j}-\sum^r_{j=1}a_{\vec{n}+\vec{e}_k,j},
\end{gather*}
which is the same as \eqref{2.2}. This f\/inishes the proof of Theorem~\ref{T1}.

\subsection{Time dependent m-OPs}\label{subsec:2.3}
Introducing the moments def\/ined by the functional \eqref{orthogonality}
\begin{gather}\label{moment}
\mu _{i,j}:=\mathcal{L}_j[x^i],\qquad j=1,\dots ,r,
\end{gather}
we have a determinant expression for m-OPs
\begin{gather}\label{detmops}
P_{\vec{n}}(x)=\frac{1}{\tau_{\vec{n}}}
\begin{vmatrix}
\mu_{0,1} & \cdots &\mu_{n_1-1,1} & \cdots &\mu_{0,r} & \cdots& \mu_{n_r-1,r} & 1\\
\mu_{1,1} & \cdots &\mu_{n_1,1} & \cdots &\mu_{1,r} & \cdots & \mu_{n_r,r} & x\\
\vdots & \cdots & \vdots & \vdots & \cdots & \vdots & \cdots & \vdots \\
\mu_{|\vec{n}|,1} & \cdots & \mu_{|\vec{n}|+n_1-1,1} & \cdots & \mu_{|\vec{n}|,r} & \cdots & \mu_{|\vec{n}|+n_r-1,r} & x^{|\vec{n}|}
\end{vmatrix},
\end{gather}
with
\begin{gather*}
\tau_{\vec{n}} =
\begin{vmatrix}
\mu_{0,1} & \cdots & \mu_{n_1-1,1} & \cdots & \mu_{0,r} & \cdots & \mu_{n_r-1,r}\\
\mu_{1,1} & \cdots & \mu_{n_1,1} & \cdots & \mu_{1,r} & \cdots & \mu_{n_r,r}\\
\vdots & \cdots & \vdots & \cdots & \vdots & \cdots & \vdots \\
\mu_{|\vec{n}|-1,1} & \cdots & \mu_{|\vec{n}|+n_1-2,1} & \cdots & \mu_{|\vec{n}|-1,r} & \cdots & \mu_{|\vec{n}|+n_r-2,r}
\end{vmatrix}.
\end{gather*}
From this expression, it is easy to f\/ind that the multi-index $\vec{n}\in \mathbb{Z}_{+}^n$ is normal if\/f~$\tau_{\vec{n}}\ne 0$ and this is assumed to hold in what follows. Using (\ref{detmops}), (\ref{rec}) and the orthogonality relation (\ref{orthogonality}), the following determinant expression of the recurrence coef\/f\/icients $\{ a_{\vec{n},j},b_{\vec{n},j}\}$ is directly verif\/ied
\begin{gather}\label{tau}
b_{\vec{n},j}=\frac{\sigma_{\vec{n}+\vec{e}_j}}{\tau_{\vec{n}+\vec{e}_j}}-\frac{\sigma_{\vec{n}}}{\tau_{\vec{n}}},\qquad a_{\vec{n},j}=\frac{\tau_{\vec{n}+\vec{e}_j}\tau_{\vec{n}-\vec{e}_j}}{\tau_{\vec{n}}^2},
\end{gather}
with
\begin{gather*}
\sigma_{\vec{n}}=
\begin{vmatrix}
\mu_{0,1} & \cdots &\mu_{n_1-1,1} & \cdots & \mu_{0,r} & \cdots & \mu_{n_r-1,r}\\
\vdots & \cdots & \vdots & \cdots & \vdots & \cdots & \vdots \\
\mu_{|\vec{n}|-2,1} & \cdots & \mu_{|\vec{n}|+n_1-3,1} & \cdots & \mu_{|\vec{n}|-2,r} & \cdots & \mu_{|\vec{n}|+n_r-3,r} \\
\mu_{|\vec{n}|,1} & \cdots & \mu_{|\vec{n}|+n_1-1,1} & \cdots & \mu_{|\vec{n}|,r} & \cdots & \mu_{|\vec{n}|+n_r-1,r}
\end{vmatrix}.
\end{gather*}

As was already mentioned, the spectral transformation plays a central role in f\/inding the corresponding integrable systems. In order to mimic the classical scheme, we will consider the spectral transformation of m-OPs f\/irst. Let us introduce the 1-parameter deformation of the moments \eqref{moment} as follows
\begin{align}\label{1pt}
\frac{d}{dt}\mu_{i,j}=-\mu_{i+1,j},\qquad j=1,\dots ,r.
\end{align}
This transformation can also be interpreted in terms of the linear
functionals \eqref{orthogonality}:
\begin{gather}\label{1pt2}
\frac{d}{dt}\mathcal{L}_j[\cdot ]=-\mathcal{L}_j[x\cdot ],\qquad j=1,\dots ,r.
\end{gather}
Notice that the 1-parameter deformation~\eqref{1pt} (or~\eqref{1pt2}) coincides with that of ordinary OPs in the case $r=1$. Using the determinant expression of m-OPs~\eqref{detmops}, the spectral transformation of m-OPs can be constructed.

\begin{Theorem}
 If the linear functionals with the condition \eqref{1pt2} $($or equivalently the mo\-ments~$\{ \mu_{i,j}\}$ with \eqref{1pt}$)$ are given,
 then the following relation for the corresponding m-OPs holds
\begin{gather}\label{diff}
\frac{d}{dt}P_{\vec{n}}(x)=\sum_{k=1}^r a_{\vec{n},k}P_{\vec{n}-\vec{e}_k}(x),
\end{gather}
where $\{ a_{\vec{n},j}\}$ are the coefficients of the recurrence relation in~\eqref{rec}.
\end{Theorem}
\begin{proof}
To begin with, introduce the notation
\begin{gather*}
\tau_{\vec{n},x}:=\big|0^{n_1},1^{n_1},\dots ,(n_1-1)^{n_1},\dots,0^{n_r}, \dots ,(n_r-1)^{n_r},x\big| \\
\hphantom{\tau_{\vec{n},x}}{} =\begin{vmatrix}
\mu_{0,1} & \mu_{1,1} & \cdots &\mu_{n_1-1,1} & \cdots & \mu_{0,r} & \cdots& \mu_{n_r-1,r} & 1\\
\mu_{1,1} & \mu_{2,1} & \cdots &\mu_{n_1+1,1} & \cdots & \mu_{1,r} & \cdots & \mu_{n_r,r} & x \\
\vdots & \vdots & \cdots & \vdots & \cdots & \vdots & \cdots & \vdots & \vdots \\
\mu_{|\vec{n}|,1} & \mu_{|\vec{n}|+1,1} & \cdots & \mu_{|\vec{n}|+n_1-1,1} & \cdots &\mu_{|\vec{n}|,r} & \cdots & \mu_{|\vec{n}|+n_r-1,r} & x^{|\vec{n}|}
\end{vmatrix}.
\end{gather*}
 Then $\tau_{\vec{n}}$ and $P_{\vec{n}}(x)$ can be rewritten as
 \begin{gather*}
 P_{\vec{n}}(x)=\frac{\tau_{\vec{n},x}}{\tau_{\vec{n}}},\qquad \tau_{\vec{n}}=\big|0^{n_1},\dots ,(n_1-1)^{n_1},\dots ,0^{n_r},\dots ,(n_r-1)^{n_r-1},e\big|,
 \end{gather*}
 where $e=(0,\dots ,0,1)^T$.
 From the relation~\eqref{1pt}, it is easy to f\/ind that $\tau_{\vec{n}}$ and $\tau_{\vec{n},x}$ are Wronskian matrices, which amounts to
\begin{gather*}
-\frac{d}{dt}\tau_{\vec{n}} =\sum_{k=1}^r\big|0^{n_1}, \dots, (n_k-2)^{n_k},n_k^{n_k},0^{n_{k+1}},\dots ,(n_r-1)^{n_r-1},e\big|, \\
-\frac{d}{dt}\tau_{\vec{n},x} =\sum_{k=1}^r\big|0^{n_1}, \dots, (n_k-2)^{n_k},n_k^{n_k},0^{n_{k+1}},\dots ,(n_r-1)^{n_r-1},x\big|.
\end{gather*}
Using these notations and relations, we can calculate the derivative of $P_{\vec{n}}(x)$:
\begin{gather}
-\frac{d}{dt}P_{\vec{n}}(x) =-\frac{d}{dt}\frac{\tau_{\vec{n},x}}{\tau_{\vec{n}}}=-\frac{\frac{d}{dt}\tau_{\vec{n},x}\tau_{\vec{n}}-\tau_{\vec{n},x}\frac{d}{dt}
\tau_{\vec{n}}}{\tau_{\vec{n}}^2}\nonumber\\
\hphantom{-\frac{d}{dt}P_{\vec{n}}(x)}{}
=-\frac{1}{\tau_{\vec{n}}^2}\sum_{k=1}^r\Bigl( \big|\dots ,(n_k-2)^{n_k},n_k^{n_k},\dots ,x\big|\big|\dots ,(n_k-2)^{n_k},(n_k-1)^{n_k},\dots ,e\big|\nonumber\\
\hphantom{-\frac{d}{dt}P_{\vec{n}}(x)=}{}
 - \big|\dots ,(n_k-2)^{n_k},(n_k-1)^{n_k},\dots ,x\big|\big|\dots ,(n_k-2)^{n_k},n_k^{n_k},\dots ,e\big| \Bigr)\nonumber\\
\hphantom{-\frac{d}{dt}P_{\vec{n}}(x)}{}
=\frac{1}{\tau_{\vec{n}}^2}\sum_{k=1}^r \big|\dots ,(n_k-2)^{n_k},\dots ,x,e\big|\big|\dots ,(n_k-2)^{n_k},(n_k-1)^{n_k},n_k^{n_k},\dots\big| \nonumber\\
\hphantom{-\frac{d}{dt}P_{\vec{n}}(x)}{}
=\sum_{k=1}^r\frac{\tau_{\vec{n}-\vec{e}_k,x}\tau_{\vec{n}+\vec{e}_k}}{\tau_{\vec{n}}^2}
=\sum_{k=1}^r\frac{\tau_{\vec{n}-\vec{e}_k}\tau_{\vec{n}+\vec{e}_k}}{\tau_{\vec{n}}^2}\frac{\tau_{\vec{n}-\vec{e}_k,x}}{\tau_{\vec{n}-\vec{e}_k}}.
\label{calculation}
\end{gather}
In the calculation of \eqref{calculation}, we have used the Pl\"{u}cker relation, a well-known identity for determinants
\begin{gather*}
|\dots ,a,b | |\dots ,c,d|-|\dots ,a,c | |\dots ,b,d|+|\dots ,a,d | |\dots ,b,c |=0,
\end{gather*}
where $a$, $b$, $c$, $d$ are arbitrary column vectors of appropriate size. Finally, comparing the result with \eqref{tau}, we arrive at \eqref{diff}. This completes the proof.
\end{proof}

We thus have obtained $r+1$ linear equations where m-OPs appear as an eigenfunction (we shall refer to this as a ``Lax set'')
\begin{gather*}
xP_{\vec{n}}(x) =P_{\vec{n}+\vec{e}_j}(x)+b_{\vec{n},j}P_{\vec{n}}(x)+\sum_{k=1}^ra_{\vec{n},k}P_{\vec{n}-\vec{e}_k}(x), \qquad j=1,\dots ,r,\nonumber\\
\frac{d}{dt}P_{\vec{n}}(x) =\sum_{k=1}^r a_{\vec{n},k}P_{\vec{n}-\vec{e}_k}(x).\label{lax}
\end{gather*}
It is straightforward to see that the compatibility condition of the Lax set~\eqref{lax} gives us the nonlinear system containing $r^2+r$ equations, even though this system apparently looks like an overdetermined system. However, with the help of the consistency relations~\eqref{diffcoeff}, the $r^2+r$ equations are reduced to $2r$ equations and the evolution of this system is thus uniquely determined. Arranging these arguments, we again arrive at the equations~\eqref{I_2.2} of Theorem~\ref{T1}:
\begin{gather}
\frac{d}{dt}a_{\vec{n},k} =a_{\vec{n},k}[b_{\vec{n}-\vec{e}_k,k}-b_{\vec{n},k}],\nonumber\\
\frac{d}{dt}b_{\vec{n},k} =\sum_{j=1}^r (a_{\vec{n},j}-a_{\vec{n}+\vec{e}_k,j}), \qquad 1 \leq k \leq r.\label{gtoda}
\end{gather}

Recalling the determinant expression of the solution \eqref{tau}, one can easily f\/ind that the coef\/f\/icients $a_{\vec{n},j}$ and $b_{\vec{n},j}$ are also expressed as follows
\begin{gather*}
b_{\vec{n},j}=\frac{\frac{d}{dt}\tau_{\vec{n}}}{\tau_{\vec{n}}}-\frac{\frac{d}{dt}\tau_{\vec{n}+\vec{e}_j}}{\tau_{\vec{n}+\vec{e}_j}},\qquad a_{\vec{n},j}=\frac{\tau_{\vec{n}+\vec{e}_j}\tau_{\vec{n}-\vec{e}_j}}{\tau_{\vec{n}}^2}.
\end{gather*}
Substituting this expression into~\eqref{gtoda}, we can can get the following statement.
\begin{Proposition}
The $\tau$-function verif\/ies the following bilinear equation
\begin{gather*}
\tau_{\vec{n}} \frac{d^2}{dt^2}\tau_{\vec{n}}=\left( \frac{d}{dt}\tau_{\vec{n}}\right)^2+\sum_{k=1}^r \tau_{\vec{n}+\vec{e}_k}\tau_{\vec{n}-\vec{e}_k}.
\end{gather*}
\end{Proposition}
This bilinear equation is also a generalization of that of the ordinary Toda equation.

\subsection{Toda chains for the diagonal m-OPs}\label{subsec:2.4}
Recall that $r$-OPs (or diagonal, or step-line m-OPs) could be recovered from the lattice of m-OPs using~(\ref{dop}). One can thus derive another integrable system, especially related to $r$-OPs, from the m-Toda lattice~(\ref{gtoda}). Here we illustrate this taking the case $r=2$ for simplicity. Let us take in \eqref{orthogonality} two dif\/ferent linear functionals $\mathcal{L}_1$, $\mathcal{L}_2$ and denote the corresponding m-OPs by~$P_{m,n}$. Then the recurrence relations~\eqref{rec} of m-OPs $P_{m,n}$ take the form
\begin{gather}
xP_{m,n}(x) =P_{m+1,n}(x)+b_{m,n,1}P_{m,n}(x)+a_{m,n,1}P_{m-1,n}(x)+a_{m,n,2}P_{m,n-1}(x),\nonumber\\
xP_{m,n}(x) =P_{m,n+1}(x)+b_{m,n,2}P_{m,n}(x)+a_{m,n,1}P_{m-1,n}(x)+a_{m,n,2}P_{m,n-1}(x),\label{rec2}
\end{gather}
and the corresponding integrable system \eqref{I_2.2} is
\begin{gather}
\dot{b}_{m,n,1} =a_{m+1,n,1}-a_{m,n,1}+a_{m+1,n,2}-a_{m,n,2},\nonumber\\
\dot{b}_{m,n,2} =a_{m,n+1,1}-a_{m,n,1}+a_{m,n+1,2}-a_{m,n,2},\nonumber\\
\dot{a}_{m,n,1} =a_{m,n,1}(b_{m,n,2}-b_{m-1,n,2}), \qquad \dot{a}_{m,n,2}=a_{m,n,2}(b_{m,n,2}-b_{m,n-1,2}),\label{gtoda2}
\end{gather}
with the contiguous relations for the initial values
\begin{gather}
b_{m,n+1,1}-b_{m,n,1}=b_{m+1,n,2}-b_{m,n,2},\nonumber\\
a_{m+1,n,1}-a_{m,n+1,1}+a_{m+1,n,2}-a_{m,n+1,2}= b_{m,n,1}b_{m+1,n,2}-b_{m,n+1,1}b_{m,n,2},\nonumber\\
a_{m,n+1,1}(b_{m-1,n,1}-b_{m-1,n,2})=a_{m,n,1}(b_{m,n,1}-b_{m,n,2}),\nonumber\\
a_{m+1,n,1}(b_{m,n-1,1}-b_{m,n-1,2})=a_{m,n,2}(b_{m,n,1}-b_{m,n,2}).\label{iv:gtoda2}
\end{gather}
We get the 2-orthogonal polynomials $\{ q_n\}$ in the following manner
\begin{gather*}
q_{2n}(x)=p_{n,n}(x),\qquad q_{2n+1}(x)=p_{n+1,n}(x).
\end{gather*}
It can easily be checked that $\{ q_n(x)\}$ satisfy the four-term recurrence relation
\begin{align}\label{rec2op}
xq_{n}(x)=q_{n+1}(x)+\alpha^{(0)}_nq_{n}(x)+\alpha^{(1)}_n q_{n-1}(x)+\alpha^{(2)}_nq_{n-2}(x),
\end{align}
where the coef\/f\/icients are
\begin{gather}
\alpha^{(0)}_{2n}=b_{n,n,1},\qquad \alpha^{(0)}_{2n+1}=b_{n+1,n,2},\nonumber\\
\alpha^{(1)}_{2n}=a_{n,n,1}+a_{n,n,2},\qquad \alpha^{(1)}_{2n+1}=a_{n+1,n,1}+a_{n+1,n,2},\nonumber\\
\alpha^{(2)}_{2n}=a_{n,n,1}(b_{n-1,n-1,1}-b_{n-1,n-1,2}),\nonumber\\
\alpha^{(2)}_{2n+1}=a_{n+1,n,2}(b_{n,n-1,1}-b_{n,n-1,2}).\label{miura}
\end{gather}
Taking all this into account, it is straightforward to f\/ind the spectral transformation of the 2-orthogonal polynomials $q_n$ as follows
\begin{gather}\label{spec2op}
\dot{q}_n(x)=-\alpha^{(1)}_n q_{n-1}(x)-\alpha^{(2)}_n q_{n-2}(x).
\end{gather}
The two equations \eqref{rec2op} and \eqref{spec2op} are exactly the Lax pair of 2-orthogonal polynomials and then the integrable system
associated with 2-orthogonal polynomials is directly derived.
\begin{Theorem} We have that the following system
\begin{gather}
\frac{d}{dt}\alpha^{(0)}_n = \alpha^{(1)}_{n+1}-\alpha^{(1)}_n,\nonumber\\
\frac{d}{dt}\alpha^{(1)}_n =\alpha^{(1)}_n\big(\alpha^{(0)}_n-\alpha^{(0)}_{n-1}\big)+\alpha^{(2)}_{n+1}-\alpha^{(2)}_n,\nonumber\\
\frac{d}{dt}\alpha^{(2)}_n =\alpha^{(2)}_n\big(\alpha^{(0)}_n-\alpha^{(0)}_{n-2}\big),\label{ktoda}
\end{gather}
is satisfied. Moreover, one can rewrite the system in the Lax form
\begin{gather*}
\frac{d}{dt}L:=[L,(L)_{-}],\qquad
L:=\begin{pmatrix}
\alpha^{(0)}_0 & 1 & & & \\
\alpha^{(1)}_1 & \alpha^{(0)}_1 & 1 & & \\
\alpha^{(2)}_2 & \alpha^{(1)} _2 & \alpha^{(0)}_2 & 1 & \\
 & \ddots & \ddots & \ddots & \ddots
\end{pmatrix},
\end{gather*}
where $(X)_-$ denotes the strictly lower part of the semi-infinite matrix $X$.
\end{Theorem}

This system is exactly the special case of the full Kostant--Toda lattice investigated in \cite{Rolania} (see also \cite{AFM2011}). It should be noted that the case $r=2$ is discussed in \cite{Rolania}, while our method is valid for the general case $r \geq 2$.
\begin{Remark}
The correspondence (\ref{miura}) is nothing but the Miura transformation from the m-Toda lattice (\ref{gtoda2}) to the special case of the
 full Kostant--Toda lattice (\ref{ktoda}).
\end{Remark}

\subsection{Example}\label{subsec:2.5}
We shall exhibit an interesting exact solution to the m-Toda lattice \eqref{gtoda2}. Let us introduce the m-OPs $P_{m,n}$ which satisfy the following multiple orthogonality relation
\begin{gather*}
\mathcal{L}_1\big[x^iP_{m,n}(x)\big] =\int _{0}^{\infty } x^{i+\delta }P_{m,n}(x)e^{-x(t+\kappa _1)}dx=0,\qquad i=0,\dots ,m-1 ,\\
\mathcal{L}_2\big[x^jP_{m,n}(x)\big] =\int _{0}^{\infty } x^{j+\delta }P_{m,n}(x)e^{-x(t+\kappa _2)}dx=0,\qquad j=0,\dots ,n-1,
\end{gather*}
where $\delta>-1$, $\kappa_1 \ne \kappa_2$ and $t>-\kappa _i$ is assumed for $i=1,2$. The corresponding m-OPs belong to the class of multiple Laguerre
polynomials of the second kind~\cite{ABVA} and they are shown~\cite{vanA2011} to satisfy the nearest-neighbor recurrence relation~(\ref{rec2}) with
\begin{gather}
b_{m,n,1} =\frac{2m+n+\delta +1}{\kappa _1+t}+\frac{n}{\kappa _2+t},\qquad
b_{m,n,2}=\frac{m+2n+\delta +1}{\kappa _2+t}+\frac{m}{\kappa _1+t}, \nonumber\\
a_{m,n,1} =\frac{m(m+n+\delta )}{(\kappa_1+t)^2}, \qquad a_{m,n,2}=\frac{n(m+n+\delta )}{(\kappa _2+t)^2}.\label{sol:gtoda2}
\end{gather}

It is easily verif\/ied that the corresponding linear functionals $\mathcal{L}_1$, $\mathcal{L}_2$ have the property~\eqref{1pt2}. Hence the coef\/f\/icients (\ref{sol:gtoda2}) directly give the special solutions to the m-Toda lattice~(\ref{gtoda2}). Furthermore, by using~(\ref{miura}), we can also obtain the corresponding solution to the special case of the full Kostant--Toda lattice \eqref{ktoda} as follows
\begin{gather*}
\alpha^{(0)}_{2n} =\frac{3n+\delta +1}{\kappa _1+t}+\frac{n}{\kappa _2+t},\qquad \alpha^{(0)}_{2n+1}=\frac{3n+\delta +2}{\kappa _2+t}+\frac{n+1}{\kappa _1+t},\\
\alpha^{(1)}_{2n} =\frac{n(2n+\delta) }{(\kappa_1 +t)^2}+\frac{n(2n+\delta) }{(\kappa_2 +t)^2 },\qquad \alpha^{(1)}_{2n+1}=\frac{(n+1)(2n+\delta+1) }{(\kappa_1 +t)^2}+\frac{n(2n+\delta+1) }{(\kappa_2 +t)^2 }, \\
\alpha^{(2)} _{2n} =(\kappa _2-\kappa _1)\frac{n(2n+\delta )(2n+\delta -1)}{(\kappa_1+t)^3 (\kappa _2 +t)},\qquad \alpha^{(2)} _{2n+1}=(\kappa _1-\kappa _2)\frac{n(2n+\delta)(2n+\delta+ 1)}{(\kappa_1+t) (\kappa _2 +t)^3}.
\end{gather*}
Interestingly, it is obvious that these solutions have a pole at $t=-\kappa_1$ and $t=-\kappa_2$, which shows the presence of singularities at certain f\/inite times.

\section{The discrete-time higher analogues of the Toda lattice}\label{sec:3}

\subsection{Discrete-time Toda equations and the q-d algorithm}\label{subsec:3.1}
In this section we recast two approaches presented in \cite{PGR1995, SNderK2011,Spiridonov2,SpZh}. The f\/irst method, which was presented in~\cite{PGR1995} and~\cite{Spiridonov2}, will then be generalized in Subsection~\ref{subsec:3.2} to the case of multiple orthogonal polynomials. As for the second one \cite{SNderK2011}, we will show in this section how to modify this to be applicable in the case of multiple orthogonal polynomials. Finally, following the discrete integrability approach proposed in \cite{BS2002}, we will adapt it to the settings in question in
Subsection~\ref{subsec:3.4}.

To start with, suppose we are given a positive measure $d\mu$ on $(0,+\infty)$ for which all the moments exist. Clearly, the measure $x^t d\mu(x)$, which is def\/ined on $(0,+\infty)$, is also positive for all $t\in{\mathbb Z}_+$. Let $P_n^{t}$ be the family of polynomials orthogonal with respect to the measure $x^t d\mu(x)$ on~$(0,+\infty)$. Introducing the moments $\mu_j$ of the measure~$d\mu$
\begin{gather*}
\mu_j=\int_0^{\infty} x^{j} d\mu(x), \qquad j=0,1,\dots,
\end{gather*}
one can easily check that the monic orthogonal polynomials $P_n^{t}$ can be presented in the following manner
\begin{gather*}
P_{n}^{t}(x)= \frac{1}{\tau_{n}^{t}}\left|\begin{matrix}
\mu_{t} & \ldots & \mu_{n+t-1} & \mu_{n+t} \\
\vdots & \vdots & \vdots & \vdots \\
\mu_{n+t-1} & \ldots & \mu_{2n+t-2} & \mu_{2n+t-1} \\
1 & \ldots & x^{n-1} & x^{n}
\end{matrix}\right|,
\end{gather*}
with the corresponding Hankel determinant
\begin{gather*}
\tau_{n+1}^{t}= \left|\begin{matrix}
\mu_{t} & \ldots & \mu_{n+t} \\
\vdots & \vdots & \vdots \\
\mu_{n+t} & \ldots & \mu_{2n+t}
\end{matrix}\right|.
\end{gather*}
It is well known that orthogonal polynomials are related by means of recurrence relations. To get those relations in the form that will be used here, let us recall the \textit{Sylvester identity}
\begin{gather*}
|A||A_{r,s;p,q}|=|A_{r;p}||A_{s;q}|-|A_{r;q}||A_{s;p}|,
\end{gather*}
where $|A|$ stands for the determinant of the $j\times j$ matrix $A$ and $A_{r,s;p,q}$ denotes the submatrix of~$A$ formed by deleting columns number~$r$, $s$ and rows number~$p$, $q$; $A_{\alpha;\beta}$ denotes the submatrix of~$A$ that is obtained from~$A$ by removing the~$\alpha$th column and~$\beta$th row.

Next, applying two dif\/ferent forms of the Sylvester identity
\begin{gather*}
|A||A_{1,n+1;1,n+1}|=|A_{1;1}||A_{n+1;n+1}|-|A_{1;n+1}||A_{n+1;1}|,\\
|A||A_{n,n+1;1,n+1}|=|A_{n;1}||A_{n+1;n+1}|-|A_{n;n+1}||A_{n+1;n+1}|
\end{gather*}
to the determinant $\tau_{n}^{t}P_{n}^{t}$ leads to the relations
\begin{gather}\label{PT_1}
P_{n+1}^{t}(x) = xP_{n}^{t+1}(x)-V_{n}^{t}P_{n}^{t}, \\
P_{n+1}^{t}(x) = xP_{n}^{t+2}(x)-W_{n}^{t}P_{n}^{t+1},\label{PT_2}
\end{gather}
where
\begin{gather*}
V_{n}^{t} = \frac{\tau_{n+1}^{t+1}\tau_{n}^{t}}{\tau_{n}^{t+1}\tau_{n+1}^{t}} ,\qquad
W_{n}^{t} = \frac{\tau_{n+1}^{t+1}\tau_{n}^{t+1}}{\tau_{n+1}^{t}\tau_{n}^{t+2}} .
\end{gather*}

The transformation \eqref{PT_1} from $P_{n}^{t}$ to $P_{n}^{t+1}$ is called the \textit{Christoffel transformation} (for instance see~\cite{BM04} that has a review of the theory of such transformations). The idea of the transformation is to construct polynomials orthogonal with respect to $x d\mu(x)$ provided
that the polynomials orthogonal with respect to $d\mu(x)$ are given. Usually, the Christof\/fel transformation appears in the context of Christof\/fel--Darboux kernels rather than the Sylvester identity and that is why this transformation is called Christof\/fel transformation. More precisely, $P_{n}^{t+1}$ can be represented by means of the Christof\/fel--Darboux kernel in the following manner
\begin{gather*}
P_{n}^{t+1}(x)=\frac{1}{P_{n}^{t}(0)}\frac{P_{n+1}^{t}(x)P_{n}^{t}(0)-P_{n+1}^{t}(0)P_{n}^{t}(x)}{x},
\end{gather*}
which is just another form of~\eqref{PT_1}.

The reciprocal to the Christof\/fel transformation is called the \textit{Geronimus transformation} and it has the following form (for instance see \cite{BM04} and \cite{SpZh})
\begin{gather}\label{GerTr}
P_{n}^{t}(x)=P_{n}^{t+1}(x)+B_{n}^tP_{n-1}^{t+1}(x).
\end{gather}
Setting
\begin{gather*}
A_n^t=-\frac{P_{n+1}^{t}(0)}{P_{n}^{t}(0)}
\end{gather*}
we see that the consistency of the Christof\/fel transformation
\begin{gather}\label{ChTr}
P_{n}^{t+1}(x)=\frac{P_{n+1}^{t}(x)+A_n^tP_{n}^{t}(x)}{x}
\end{gather}
and the Geronimus transformation \eqref{GerTr} leads to the q-d algorithm or, which is equivalent, to the discrete-time Toda equation~\eqref{dtoda}~\cite{PGR1995}. Indeed, on the one hand substituting~\eqref{ChTr} to~\eqref{GerTr} gives
\begin{gather*}
xP_{n}^{t}(x)=P_{n+1}^{t}(x)+\big(A_n^t+B_n^t\big)P_{n}^{t}(x)+A_{n-1}^tB_n^tP_{n-1}^{t}.
\end{gather*}
On the other hand, if we plug \eqref{GerTr} into \eqref{ChTr} we get
\begin{gather*}
xP_{n}^{t+1}(x)=P_{n+1}^{t+1}(x)+\big(A_n^t+B_{n+1}^t\big)P_{n}^{t+1}(x)+A_{n}^tB_n^tP_{n-1}^{t+1}.
\end{gather*}
Now, comparing the corresponding coef\/f\/icients leads to \eqref{dtoda}:
\begin{gather*}
 A_{n}^{t+1}+B_n^{t+1}=A_{n}^t+B_{n+1}^t, \qquad
 A_{n-1}^{t+1}B_n^{t+1}=A_n^tB_n^t, \qquad t, n\in{\mathbb Z}_+.
\end{gather*}
In fact, this idea is the f\/irst scheme that we are going to extend to the setting of multiple orthogonal polynomials in the next subsection.

The second approach we mention here was presented in \cite{SNderK2011} for orthogonal polynomials. The authors propose to use the transformation \eqref{PT_2} rather than~\eqref{GerTr} and they show that the relations~\eqref{PT_1} and~\eqref{PT_2} give a Lax pair for the discretization of the Toda chain. Roughly speaking, the consistency of \eqref{PT_1} and \eqref{PT_2} leads to a discrete zero curvature condition (for more information about discrete integrability and zero curvature conditions see \cite{Adler2001,BS2002,PGR1995,SNderK2011}).

Now let us quickly see how it works. At f\/irst, introduce the wave function
\begin{gather*}
\widetilde{\Psi}_{n,t}(x)=\big(P^{t}_n(x),P^{t+1}_n(x)\big)^{\top}.
\end{gather*}
Next, using \eqref{PT_1} and \eqref{PT_2} we derive
\begin{gather}\label{DisToda1}
\widetilde{\Psi}_{n+1,t}=\widetilde{L}_{n,t}\widetilde{\Psi}_{n,t},\qquad
\widetilde{\Psi}_{n,t+1}=\widetilde{M}_{n,t}\widetilde{\Psi}_{n,t},
\end{gather}
where the transition matrices are def\/ined as follows
\begin{gather*}
\widetilde{L}_{n,t} = \left( \begin{matrix}
-V_{n}^{t} & x \\
-V_{n}^{t} & x + W_{n}^{t} -V_{n}^{t+1}
\end{matrix} \right), \qquad
\widetilde{M}_{n,t} = \frac{1}{x}\left( \begin{matrix}
0 & x \\
- V_{n}^{t} & x+ W_{n}^{t}
\end{matrix} \right).
\end{gather*}
The consistency of the linear systems~\eqref{DisToda1} is then equivalent to the zero curvature condition
\begin{gather}\label{DiscreteToda}
0 = \widetilde{L}_{n,t+1}\widetilde{M}_{n,t} - \widetilde{M}_{n+1,t} \widetilde{L}_{n,t},
\end{gather}
which can be simplif\/ied to the quotient-dif\/ference scheme
\begin{gather}\label{QDToda}
V_{n}^{t+2} + W_{n+1}^{t} = V_{n+1}^{t} + W_{n}^{t+1}, \qquad
W_{n}^{t} V_{n+1}^{t} = V_{n}^{t+1} W_{n+1}^{t},
\end{gather}
which, as was already mentioned, can be considered as the discrete-time Toda equa\-tion~\cite{PGR1995,SNderK2011}. While on the subject, let us mention that the quotient-dif\/ference scheme, along with many other relations between orthogonal polynomials, naturally occur in the context of Pad\'e tables~\cite{Gragg}.

Now we are in the position to modify the approach we just recalled. The reason to do that is the fact that instead of having \eqref{PT_1} and \eqref{PT_2} one is usually given the three-term recurrence relation
\begin{gather}\label{3termT}
x P_{n}^{t}(x)= P_{n+1}^{t}(x)+b_{n}^{t}P_{n}^{t}(x)+a_{n}^{t}P_{n-1}^{t}(x).
\end{gather}
At the same time, it is not so hard to check that the combination of~\eqref{PT_1} and~\eqref{PT_2} leads to the monic version of the three-term recurrence relation (see~\cite{SNderK2011})
\begin{gather*}
x P_{n}^{t}= P_{n+1}^{t}+\big(V_{n}^{t}+V_{n-1}^{t+1}-W_{n}^{t}\big)P_{n}^{t}+\big(V_{n-1}^{t+1}-W_{n-1}^{t}\big)V_{n-1}^{t}P_{n-1}^{t}.
\end{gather*}
Hence, we also have formulas for the coef\/f\/icients $a_{n}^{t}$ and $b_{n}^{t}$ in terms of Hankel determinants
\begin{gather}\label{OPRecFromDet}
a_{n}^{t}=\big(V_{n-1}^{t+1}-W_{n-1}^{t}\big)V_{n-1}^{t},\qquad b_{n}^{t}=V_{n}^{t}+V_{n-1}^{t+1}-W_{n}^{t}.
\end{gather}
Once we have the coef\/f\/icients of the three-term recurrence relation, it is in many cases a simple task to reconstruct the coef\/f\/icients of the Christof\/fel transformation (we can use either the determinant formula given in~\eqref{PT_1} or the formula in terms of the polynomials based on~\eqref{ChTr}). Thus, in order to f\/ind the Lax pair, it is preferable to use~\eqref{3termT} and~\eqref{PT_1}.
\begin{Proposition} \label{LaxPairToda} Let us consider the following vector-valued wave function
\begin{gather*}
{\Psi}_{n,t}(x)=\big(P^{t}_n(x),P^{t}_{n-1}(x)\big)^{\top}.
\end{gather*}
Then the corresponding transition matrices are
\begin{gather}\label{LP1}
L_{n,t} = \left( \begin{matrix}
b_n^{t}-x & a_n^{t} \\
1 & 0
\end{matrix} \right), \\
\label{LP2} M_{n,t} = \frac{1}{x}\left( \begin{matrix}
b_n^{t}+V_{n}^{t}-x & a_n^{t} \\
 1 & V_{n-1}^{t}
\end{matrix} \right),
\end{gather}
and they give another Lax pair for the discrete time Toda equation~\eqref{QDToda}.
\end{Proposition}
\begin{proof}
To see that the statement holds, we notice that, after some manipulations with~\eqref{3termT} and~\eqref{PT_1}, one can get the following equalities
\begin{gather*}
{\Psi}_{n+1,t}=L_{n,t}{\Psi}_{n,t},\qquad {\Psi}_{n,t+1}=M_{n,t}{\Psi}_{n,t},
\end{gather*}
where the transition matrices $L_{n,t}$ and $M_{n,t}$ are given by \eqref{LP1} and \eqref{LP2}, respectively.
Next, since we have the relation
\begin{gather*}
\widetilde{\Psi}_{n,t}(x)=
\left( \begin{matrix}
1& 0 \\
\frac{b_n^{t}+V_{n}^{t}}{x}-1 & \frac{a_n^{t}}{x}
\end{matrix} \right){\Psi}_{n,t}(x),
\end{gather*}
it is clear that
\begin{gather*}
0 = {L}_{n,t+1}{M}_{n,t} - {M}_{n+1,t} {L}_{n,t}
\end{gather*}
is equivalent to \eqref{DiscreteToda} and, in turn, reduces to~\eqref{QDToda}.
\end{proof}

\subsection{Christof\/fel and Geronimus transformations for m-OPs}~\label{subsec:3.2}
In this section, we will use the f\/irst method from the previous subsection to get an integrable discretization of the m-Toda lattice~\eqref{gtoda}. As for the discretization of integrable systems, many techniques have been proposed and investigated (see for the details, e.g., \cite{Grammaticos,Suris}). Nonetheless, it is quite convenient to construct the discretization by means of the discrete spectral transformation of m-OPs as was done in \cite{PGR1995,Spiridonov2,SpZh}.

Let us work on the discrete spectral transformations of m-OPs, a~mapping from m-OPs to another m-OPs. It is not so dif\/f\/icult to get to a generalization of the Christof\/fel transformation.

\begin{Proposition}\label{CTforMPh}
Let $\{ P_{\vec{n}}^t(x) \}$ be m-OPs at some time $t$ with respect to $\mathcal{L}_1^t,\dots, \mathcal{L}_r^t$ and define the new sequence of
 polynomials $\{ P_{\vec{n},j}^{t+1}(x)\}$ for $j=1,\dots ,r $ by
\begin{gather}\label{mct}
P_{\vec{n},j}^{t+1}(x)=\frac{1}{x-\lambda _t}\big( P_{\vec{n}+\vec{e}_j}^t(x)+A_{\vec{n},j}^tP_{\vec{n}}^t(x) \big),\qquad A_{\vec{n},j}^t=-\frac{P_{\vec{n}+\vec{e}_j}^t(\lambda _t)}{ P_{\vec{n}}^t(\lambda _t) },
\end{gather}
where $\lambda_t \in \{ z \in \mathbb{R} \,|\, P_{\vec{n}}^t(z) \ne 0,\, \textrm{for all}\,\, \vec{n} \in \mathbb{Z}_{+}^r \}$. Then
\begin{gather}\label{uniqueness}
P_{\vec{n},1}^{t+1}(x)=\cdots =P_{\vec{n},r}^{t+1}(x) = P_{\vec{n}}^{t+1}(x)
\end{gather}
holds and $\{ P_{\vec{n}}^{t+1}(x) \}$ are again m-OPs with respect to the new linear functional $\mathcal{L}_1^{t+1},\dots ,\mathcal{L}_r^{t+1}$ defined by
\begin{gather}\label{d1pt}
\mathcal{L}_j^{t+1}[\,\cdot \,]:=\mathcal{L}_j^t[(x-\lambda _t)\,\cdot\, ],\qquad j=1,\dots ,r.
\end{gather}
\end{Proposition}

\begin{proof}
From the multiple orthogonality, it is easy to verify
\begin{gather*}
\mathcal{L}_i^{t+1}\big[x^kP_{\vec{n},j}^{t+1}(x)\big]
 =\mathcal{L}_i^t\big[x^k(x-\lambda _t)P_{\vec{n},j}^{t+1}(x)\big]
 =\mathcal{L}_i^t\big[x^k\big( P_{\vec{n}+\vec{e}_j}^t(x)+A_{\vec{n},j}^t P_{\vec{n},j}^t(x)\big)\big] \\
\hphantom{\mathcal{L}_i^{t+1}\big[x^kP_{\vec{n},j}^{t+1}(x)\big]}{} =0,\qquad k=0,1,\dots ,n_i-1,
\end{gather*}
for all $i,j=1,\dots ,r$. Then the uniqueness of monic m-OPs gives us~\eqref{uniqueness} and this completes the proof.
\end{proof}

In case $r=1$, the transformation (\ref{mct}) coincides with the Christof\/fel transformation for OPs. We shall refer to the transformation (\ref{mct}) as the Christof\/fel transformation for m-OPs. Setting the initial time $t=0$ and iterating the Christof\/fel transformation
\begin{gather}
P_{\vec{n}}^{t+1}(x)=\frac{1}{x-\lambda _t}\big( P_{\vec{n}+\vec{e}_j}^t(x)+A_{\vec{n},j}^tP_{\vec{n}}^t(x) \big),\nonumber\\
A_{\vec{n},j}^t=-\frac{P_{\vec{n}+\vec{e}_j}^t(\lambda _t)}{ P_{\vec{n}}^t(\lambda _t) },\qquad j=1,\dots ,r,\label{mct2}
\end{gather}
we can obtain the chain of m-OPs
\begin{gather*}
\big\{ P_{\vec{n}}^0(x)\big\} \rightarrow \big\{ P_{\vec{n}}^1(x)\big\} \rightarrow \cdots \rightarrow \big\{ P_{\vec{n}}^t(x)\big\} \rightarrow \big\{ P_{\vec{n}}^{t+1}(x)\big\} \rightarrow \cdots .
\end{gather*}
\begin{Remark}
We can derive the nonlinear equations from the compatibility condition of the relations~\eqref{mct2} themselves
\begin{gather}\label{mkp}
A_{\vec{n},i}^{t+1}-A_{\vec{n},j}^{t+1}=A_{\vec{n}+\vec{e_j},i}^t-A_{\vec{n}+\vec{e}_i,j}^t,\qquad
A_{\vec{n},i}^tA_{\vec{n}+\vec{e}_i,j}^t=A_{\vec{n},j}^tA_{\vec{n}+\vec{e}_j,i}^t.
\end{gather}
\end{Remark}

As for the discrete time m-Toda lattice, the Geronimus transformation, which in a way is the reciprocal to the Christof\/fel transformation, plays a key role.
The Geronimus transformation for m-OPs is given in the following statement.
\begin{Theorem}
The sequence of m-OPs $\{ P_{\vec{n}}^{t}(x) \}$ connected by the relation~\eqref{mct2} satisfies the following contiguous relations
\begin{gather}\label{mgt}
P_{\vec{n}}^{t}(x)=P_{\vec{n}}^{t+1}(x)+\sum_{k=1}^rB_{\vec{n},k}^tP_{\vec{n}-\vec{e}_k}^{t+1}(x),\qquad B_{\vec{n},k}^t=-a_{\vec{n},k}^t\frac{P_{\vec{n}-\vec{e}_k}^t(\lambda _t)}{P_{\vec{n}}^t(\lambda _t)},
\end{gather}
where $a_{\vec{n},k}^t$ are the coefficients of the nearest neighbor recurrence relation
\begin{gather}\label{trec}
xP_{\vec{n}}^t(x)=P_{\vec{n}+\vec{e}_j}^t(x)+b_{\vec{n},j}^tP_{\vec{n}}^t(x)+\sum_{k=1}^ra_{\vec{n},k}^tP_{\vec{n}-\vec{e}_k}^t(x), \qquad j=1,\dots ,r.
\end{gather}
\end{Theorem}
\begin{proof}
First, consider the polynomial $P_{\vec{n}}^t(x)-P_{\vec{n}}^{t+1}(x)$, which is a polynomial of degree $|\vec{n}|-1$. From the multiple orthogonality relation and \eqref{d1pt}, we can easily check
\begin{gather}\label{change}
\mathcal{L}_i^{t+1}\big[x^jP_{\vec{n}}^t(x)\big] =\mathcal{L}_i^t\big[(x-\lambda _t)x^jP_{\vec{n}}^t(x)\big]=0,\qquad j=0,\dots, n_i-2.
\end{gather}
This readily shows $P_{\vec{n}}^t(x)-P_{\vec{n}}^{t+1}(x)$ is orthogonal to all polynomials of degree less than $n_j-1$ with respect to the linear functional $\mathcal{L}_j^{t+1}$. Hence, we can write $P_{\vec{n}}^t(x)-P_{\vec{n}}^{t+1}(x)$ as a linear combination of the polynomials $P^{t+1}_{\vec{n}-\vec{e}_j}$, $j=1,\dots ,r$, which form a basis for the linear space of all polynomials of degree less than $|\vec{n}|$ and satisfy the multiple orthogonality conditions
\begin{gather*}
\mathcal{L}_j^{t+1}\big[x^kP_{\vec{n}}^{t+1}(x)\big]=0,\qquad j=1,\dots ,r,
\end{gather*}
for $k\le n_j-2$. We write $P_{\vec{n}}^t(x)-P_{\vec{n}}^{t+1}(x)=\sum\limits_{k=1}^r B_{\vec{n},k}^tP_{\vec{n}-\vec{e}_k}^{t+1}(x)$.
From \eqref{change}, we can easily f\/ind
\begin{gather}\label{p1}
B_{\vec{n},k}^t=\frac{\mathcal{L}_k^{t+1}[x^{n_k-1}P_{\vec{n}}^t(x)]}{\mathcal{L}_k^{t+1}[x^{n_k-1}P_{\vec{n}-\vec{e}_k}^{t+1}(x)]},
\qquad k=1,\dots ,r.
\end{gather}
Here, some calculations show
\begin{gather}\label{p2}
\mathcal{L}_k^{t+1}\big[x^{n_k-1}P_{\vec{n}}^t(x)\big]=\mathcal{L}_k^{t}\big[(x-\lambda _t)x^{n_k-1}P_{\vec{n}}^t(x)\big]=\mathcal{L}_k^{t}\big[x^{n_k}P_{\vec{n}}^t(x)\big],
\end{gather}
and
\begin{gather}
\mathcal{L}_k^{t+1}\big[x^{n_k-1}P_{\vec{n}-\vec{e}_k}^{t+1}(x)\big]
 =\mathcal{L}_k^{t}\big[(x-\lambda _t)x^{n_k-1}P_{\vec{n}-\vec{e}_k}^{t+1}(x)\big]
 =\mathcal{L}_k^{t}\big[x^{n_k-1}\big(P_{\vec{n}}^t(x)+A_{\vec{n},k}^tP_{\vec{n}-\vec{e}_k}^t(x)\big)\big]\nonumber\\
 \hphantom{\mathcal{L}_k^{t+1}\big[x^{n_k-1}P_{\vec{n}-\vec{e}_k}^{t+1}(x)\big]}{}
 =A_{\vec{n},k}^t\mathcal{L}_k^t\big[x^{n_k-1}P_{\vec{n}-\vec{e}_k}^t(x)\big].\label{p3}
\end{gather}
We can also calculate the coef\/f\/icients of \eqref{trec} from the multiple orthogonality relation
\begin{gather}\label{p4}
a_{\vec{n},k}^t=\frac{\mathcal{L}_k^t[x^{n_k}P_{\vec{n}}^t(x)]}{\mathcal{L}_k^t[x^{n_k-1}P_{\vec{n}-\vec{e}_k}^t(x)]}.
\end{gather}
Combining \eqref{p1}--\eqref{p4}, we f\/inally arrive at the following result:
\begin{gather*}
B_{\vec{n}}^t=\frac{a_{\vec{n},k}^t}{A_{\vec{n},k}^t}=-a_{\vec{n},k}^t\frac{P_{\vec{n}-\vec{e}_k}^t(\lambda _t)}{P_{\vec{n}}^t(\lambda _t)},
\qquad k=1,\dots ,r.
\end{gather*}
This completes the proof.
\end{proof}

\begin{Remark}
Although we can formally derive all the relations in Section~\ref{sec:3} for $\lambda_t$ that changes with the discrete time $t$, we are only concerned with the case $\lambda_t=\lambda$ and, particularly, $\lambda_t=0$. The reason is that the basis for our construction of integrable systems is a perfect system of measures $d\mu_1, \dots, d\mu_r$ and it has to be perfect at all times. In other words, we need to have all the multiple orthogonal polynomials to exist for any multi-index for all values of the time. However, it is still an open question when the system remains perfect under Christof\/fel or Geronimus transformations. Moreover, the only examples we know at the moment correspond to the case $\lambda_t=0$, which can be easily modif\/ied to $\lambda_t=\lambda$ (see Subsection~\ref{subsec:3.5}).
\end{Remark}

From \eqref{mct2} and \eqref{mgt}, we can reproduce the nearest neighbor recurrence relation \eqref{trec} and the coef\/f\/icients can explicitly be written in terms of $\big\{ A_{\vec{n},j}^t,B_{\vec{n},j}^t\big\}$
\begin{gather}\label{tcoeff}
b_{\vec{n},j}^t=A_{\vec{n},j}^t+\sum_{k=1}^rB_{\vec{n},k}^t+\lambda _t,\qquad a_{\vec{n},j}^t=A_{\vec{n}-\vec{e}_j,j}^tB_{\vec{n},j}^t,\qquad j=1,\dots ,r.
\end{gather}
Substituting~\eqref{tcoeff} into~\eqref{diffcoeff} and also using the relation~\eqref{mkp}, we obtain, after some calculations and simplif\/ications,
the contiguous relations for $\{ A_{\vec{n},j}^t,B_{\vec{n},j}^t\} $.
\begin{Corollary}
If we put $\lambda_t=0$, the coefficients $\big\{ A_{\vec{n},j}^t,B_{\vec{n},j}^t\big\}$ in~\eqref{mct2} and \eqref{mgt} satisfy the following difference equations on~$\vec{n}$
\begin{gather}
A_{\vec{n},i}^tA_{\vec{n}+\vec{e}_i,j}^t=A_{\vec{n},j}^tA_{\vec{n}+\vec{e}_j,i}^t, \nonumber\\
A_{\vec{n}+\vec{e}_j,i}^t-A_{\vec{n}+\vec{e_i},j}^t+A_{\vec{n},j}^t-A_{\vec{n},i}^t=\sum_{k=1}^r\big(B_{\vec{n}+\vec{e}_i,k}^t-B_{\vec{n}+\vec{e}_j,k}^t\big),\nonumber\\
\frac{B_{\vec{n},i}^t}{B_{\vec{n}+\vec{e}_j,i}^t}=\frac{A_{\vec{n},j}^t-A_{\vec{n}+\vec{e}_j-\vec{e}_i,i}^t}{A_{\vec{n},j}^t-A_{\vec{n},i}^t}\label{div}
\end{gather}
for all $t$ and $i,j=1,\dots ,r$.
\end{Corollary}

From the Christof\/fel and Geronimus transformations for m-OPs, we obtain the discrete Lax set for which the m-OPs appear as their eigenfunctions
\begin{gather}
(x-\lambda _t)P_{\vec{n}}^{t+1}(x)=P_{\vec{n}+\vec{e}_j}^t(x)+A_{\vec{n},j}^tP_{\vec{n}}^t(x),\qquad j=1,\dots ,r,\nonumber\\
P_{\vec{n}}^t(x)=P_{\vec{n}}^{t+1}(x)+\sum_{k=1}^rB_{\vec{n},k}^tP_{\vec{n}-\vec{e}_k}^{t+1}(x).\label{dlax}
\end{gather}
Next, the compatibility condition for the relations in~\eqref{dlax} gives us $r^2+r$ equations and, therefore, the obtained system seems overdetermined.
However, if we take into account the contiguous relation~\eqref{div}, then these $r^2+r$ equations are reduced to~$2r$ equations and the evolution is uniquely determined. Summing up these arguments, we get to Theorem~\ref{T2}:
\begin{gather}
A_{\vec{n},j}^{t+1}+\sum_{k=1}^rB_{\vec{n},k}^{t+1}=A_{\vec{n},j}^t+\sum_{k=1}^rB_{\vec{n}+\vec{e_j},k}^t,\nonumber\\
A_{\vec{n}-\vec{e}_j,j}^{t+1}B_{\vec{n},j}^{t+1}=A_{\vec{n},j}^tB_{\vec{n},j}^t,
\qquad 1\leq j \leq r.\label{dgtoda}
\end{gather}

\begin{Remark}
As with the continuous time m-Toda lattice, we see here that the system (\ref{dgtoda}) is solvable (integrable) only if the initial values are chosen so that \eqref{div} are satisf\/ied.
\end{Remark}
 One can also see that in the case $r=1$ the system \eqref{dgtoda} coincides with the discrete time Toda lattice \eqref{dtoda}. Therefore, it is natural to call (\ref{dgtoda}) a \textit{discrete multiple Toda} (dm-Toda) lattice.

\begin{Remark}
The dm-Toda lattice (\ref{dgtoda}) is exactly the discrete analogue of the continuous time m-Toda lattice~(\ref{gtoda}).
Indeed, let us introduce the new variables $a_{\vec{n},j}(t)$ and $b_{\vec{n},j}(t)$ by the following relation
\begin{gather}\label{var}
\lambda _t= \frac{1}{\delta},\qquad A_{\vec{n},j}^t=b_{\vec{n},j}(t\delta )-\frac{1}{\delta },\qquad B_{\vec{n},j}^t=\delta a_{\vec{n},j}(t\delta ).
\end{gather}
Then, substituting \eqref{var} into \eqref{dgtoda}, we can obtain the following equations
\begin{gather}
\frac{b_{\vec{n},j}(t\delta +\delta )-b_{\vec{n},j}(t\delta)}{\delta }=\sum_{k=1}^r a_{\vec{n}+\vec{e}_j,k}(t\delta )-a_{\vec{n},k}(t\delta +\delta),\nonumber\\
\frac{a_{\vec{n},j}(t\delta +\delta )-a_{\vec{n},j}(t\delta )}{\delta }=b_{\vec{n},j}(t\delta )a_{\vec{n},j}(t\delta )-b_{\vec{n}-\vec{e}_j,j}(t\delta+\delta )a_{\vec{n},j}(t\delta +\delta ).\label{dgtoda2}
\end{gather}
If we take $t\delta \rightarrow t$ and the continuous limit $\delta \rightarrow 0$, it is straightforward to see that the equa\-tions~\eqref{dgtoda2} go to the m-Toda lattice \eqref{gtoda}.
\end{Remark}

\begin{Remark}
We can also verify that the contiguous relation \eqref{div} reduce to \eqref{diffcoeff} in the continuous limit after some careful calculations and simplif\/ications.
\end{Remark}

In the previous section, we have seen that the m-Toda lattice~\eqref{gtoda} admits the determinant solution~\eqref{tau} with the dispersion relation~\eqref{1pt}. We shall now give the determinant solution to the dm-Toda lattice~\eqref{dgtoda}. Let us introduce the $\tau$-function $\tau_{\vec{n}}^t$ def\/ined by
\begin{gather*}
\tau_{\vec{n}}^t:
=\begin{vmatrix}
\mu_{0,1}^t & \mu_{1,1}^t & \cdots &\mu_{n_1-1,1}^t & \cdots & \mu_{0,r}^t & \cdots& \mu_{n_r-1,r}^t \\
\mu_{1,1}^t & \mu_{2,1}^t & \cdots &\mu_{n_1+1,1}^t & \cdots & \mu_{1,r}^t & \cdots & \mu_{n_r,r}^t \\
\vdots & \vdots & \cdots & \vdots & \cdots & \vdots & \cdots & \vdots \\
\mu_{|\vec{n}|-1,1}^t & \mu_{|\vec{n}|,1}^t & \cdots & \mu_{|\vec{n}|+n_1-2,1}^t & \cdots &\mu_{|\vec{n}|-1,r}^t & \cdots & \mu_{|\vec{n}|+n_r-2,r}^t
\end{vmatrix},
\end{gather*}
where $\mu_{i,j}^t:=\mathcal{L}_j^t[x^i]$. From \eqref{d1pt}, we get the following relation
\begin{gather*}
\mu_{i,j}^{t+1}=\mu_{i+1,j}^t-\lambda _t\mu_{i,j}^t,\qquad j=1,\dots ,r.
\end{gather*}
From the determinant expression of the m-OPs \eqref{detmops} we can get by means of elementary transformations of determinants
\begin{gather}
P_{\vec{n}}^t(\lambda _t) =\frac{1}{\tau_{\vec{n}}^t}
\begin{vmatrix}
\mu_{0,1}^t & \cdots &\mu_{n_1-1,1}^t & \cdots &\mu_{0,r}^t & \cdots& \mu_{n_r-1,r}^t & 1 \\
\mu_{0,1}^{t+1} & \cdots &\mu_{n_1-1,1}^{t+1} & \cdots &\mu_{0,r}^t & \cdots& \mu_{n_r-1,r}^{t+1} & 0 \\
\mu_{1,1}^{t+1} & \cdots &\mu_{n_1,1}^{t+1} & \cdots & \mu_{1,r}^t & \cdots & \mu_{n_r,r}^{t+1} & 0 \\
\vdots & \cdots & \vdots & \cdots & \vdots & \cdots & \vdots & \vdots \\
\mu_{|\vec{n}|-1,1}^{t+1} & \cdots & \mu_{|\vec{n}|+n_1-2,1}^{t+1} & \cdots &\mu_{|\vec{n}|-1,r}^{t+1} & \cdots & \mu_{|\vec{n}|+n_r-2,r}^{t+1} & 0
\end{vmatrix}\nonumber\\
\hphantom{P_{\vec{n}}^t(\lambda _t)}{} =(-1)^{|\vec{n}|+1}\frac{\tau_{\vec{n}}^{t+1}}{\tau_{\vec{n}}^t}.\label{ddetmops}
\end{gather}
Combining \eqref{tau}, \eqref{mct}, \eqref{mgt} and \eqref{ddetmops}, we arrive at the determinant expression of the solution:
\begin{gather}\label{dep-var}
A_{\vec{n},j}^t=\frac{\tau_{\vec{n}+\vec{e}_j}^{t+1}\tau_{\vec{n}}^t}{\tau_{\vec{n}+\vec{e}_j}^t\tau_{\vec{n}}^{t+1}},\qquad B_{\vec{n},j}^t=\frac{\tau_{\vec{n}-\vec{e}_j}^{t+1}\tau_{\vec{n}+\vec{e}_j}^t}{\tau_{\vec{n}}^t\tau_{\vec{n}}^{t+1}},\qquad j=1,\dots ,r.
\end{gather}

We shall consider the bilinear equations for dm-Toda lattice \eqref{dgtoda} for the case $\lambda_t =\lambda$. From~\eqref{mkp}, the dependent variable $A_{\vec{n},j}^t $ obeys a discrete KP equation, which reduces to the following Hirota--Miwa equation
\begin{gather}\label{hirota-miwa}
\tau_{\vec{n}+\vec{e}_i+\vec{e}_j}^{t}\tau_{\vec{n}}^{t+1}-\tau_{\vec{n}+\vec{e}_i}^{t}\tau_{\vec{n}+\vec{e}_j}^{t+1}
+\tau_{\vec{n}+\vec{e}_j}^{t}\tau_{\vec{n}+\vec{e}_i}^{t+1}=0,\qquad i\ne j.
\end{gather}
Using the Hirota--Miwa equation and substituting \eqref{dep-var} into \eqref{dgtoda}, we obtain another bilinear equation of $\tau_{\vec{n}}^t$:
\begin{gather}\label{dbilinear}
\tau_{\vec{n}}^{t+1}\tau_{\vec{n}}^{t-1}=\big(\tau_{\vec{n}}^t\big)^2+\sum_{k=1}^r \tau_{\vec{n}+\vec{e}_k}^{t-1}\tau_{\vec{n}-\vec{e}_k}^{t+1},
\end{gather}
which is the multiple generalization of the bilinear equation of the ordinary discrete Toda lattice.
Indeed we can derive the dm-Toda lattice~\eqref{dgtoda} from these bilinear equations. Using~\eqref{dbilinear}, one has
\begin{gather*}
\left(\frac{\tau_{\vec{n}+\vec{e}_j}^{t+1}}{\tau_{\vec{n}}^{t+1}} \right)^2 = \frac{\tau_{\vec{n}+\vec{e}_j}^{t+2}\tau_{\vec{n}+\vec{e}_j}^t-\tau_{\vec{n}+2\vec{e}_j}^t \tau_{\vec{n}}^{t+2}-\sum\limits_{k\ne j} \tau_{\vec{n}+\vec{e}_j+\vec{e}_k}^{t}\tau_{\vec{n}+\vec{e}_j-\vec{e}_k}^{t+2}}{\tau_{\vec{n}}^{t+2}\tau_{\vec{n}}^t-\tau_{\vec{n}+\vec{e}_j}^t \tau_{\vec{n}-\vec{e}_j}^{t+2}-\sum\limits_{k\ne j} \tau_{\vec{n}+\vec{e}_k}^{t}\tau_{\vec{n}-\vec{e}_k}^{t+2}},
\end{gather*}
which can be rewritten as
\begin{gather*}
A_{\vec{n},j}^{t+1}+B_{\vec{n},j}^{t+1}+
\sum_{k\ne j}\frac{\tau_{\vec{n}+\vec{e_k}}^t\tau_{\vec{n}+\vec{e}_j}^{t+1}}{\tau_{\vec{n}+\vec{e}_j}^t \tau_{\vec{n}+\vec{e}_k}^{t+1}}B_{\vec{n},k}^{t+1}
=A_{\vec{n},j}^t+B_{\vec{n}+\vec{e}_j,j}^t
+\sum_{k\ne j}\frac{\tau_{\vec{n}}^{t+1}\tau_{\vec{n}+\vec{e}_j-\vec{e}_k}^{t+2}}{\tau_{\vec{n}}^{t+2} \tau_{\vec{n}+\vec{e}_j-\vec{e}_k}^{t+1}}B_{\vec{n}+\vec{e}_j,k}^{t}.
\end{gather*}
For all $k\ne j$, one also has
\begin{gather*}
B_{\vec{n},k}^{t+1}-\frac{\tau_{\vec{n}+\vec{e_k}}^t\tau_{\vec{n}+\vec{e}_j}^{t+1}}{\tau_{\vec{n}+\vec{e}_j}^t \tau_{\vec{n}+\vec{e}_k}^{t+1}}B_{\vec{n},k}^{t+1}=B_{\vec{n}+\vec{e}_j,k}^t-\frac{\tau_{\vec{n}}^{t+1}
\tau_{\vec{n}+\vec{e}_j-\vec{e}_k}^{t+2}}{\tau_{\vec{n}}^{t+2} \tau_{\vec{n}+\vec{e}_j-\vec{e}_k}^{t+1}}B_{\vec{n}+\vec{e}_j,k}^{t},
\end{gather*}
which is easily verif\/ied from \eqref{hirota-miwa} and \eqref{dep-var}. Comparing the above two relations we thus have the dm-Toda lattice \eqref{dgtoda}. In other words, the dm-Toda lattice is a consequence of the two bilinear equations \eqref{hirota-miwa} and \eqref{dbilinear}.

\subsection{The discrete-time integrable system and diagonal m-OP}\label{subsec:3.3}
Here we will introduce the Miura transformation from the dm-Toda lattice to the discrete integrable system associated with $r$-orthogonal polynomials, as was done in the previous section for the continuous-time m-Toda lattice. For simplicity, we only consider the case $r=2$ and the autonomous case $\lambda_t=\lambda$. We denote the corresponding m-OPs by~$\{ P_{m,n}^t\} $ and we write the discrete Lax set of m-OPs as follows
\begin{gather}
(x-\lambda _t)P_{m,n}^{t+1}(x)=P_{m+1,n}^t(x)+A_{m,n,1}^tP_{m,n}^t(x),\nonumber\\
(x-\lambda _t)P_{m,n}^{t+1}(x)=P_{m,n+1}^t(x)+A_{m,n,2}^tP_{m,n}^t(x),\nonumber\\
P_{m,n}^{t}(x)=P_{m,n}^{t+1}(x)+B_{m,n,1}^tP_{m-1,n}^{t+1}(x)+B_{m,n,2}^tP_{m,n-1}^{t+1}(x).\label{d2lax}
\end{gather}
Thus the corresponding dm-Toda equation can be represented by
\begin{gather}
A_{m,n,1}^t+B_{m+1,n,1}^t+B_{m+1,n,2}^t=A_{m,n,1}^{t+1}+B_{m,n,1}^{t+1}+B_{m,n,2}^{t+1}, \nonumber\\
A_{m,n,2}^{t}+B_{m,n+1,1}^{t}+B_{m,n+1,2}^{t}=A_{m,n,2}^{t+1}+B_{m,n,1}^{t+1}+B_{m,n,2}^{t+1},\nonumber\\
A_{m,n,1}^tB_{m,n,1}^t=A_{m-1,n,1}^{t+1}B_{m,n,1}^{t+1},\qquad A_{m,n,2}^{t}B_{m,n,2}^{t}=A_{m,n-1,2}^{t+1}B_{m,n,2}^{t+1},\label{d2gtoda}
\end{gather}
with the contiguous relations of the initial values
\begin{gather*}
A_{m,n,1}^0A_{m+1,n,2}^0=A_{m,n+1,1}^0A_{m,n,2}^0,\\
A_{m,n+1,1}^0-A_{m,n,1}^0-A_{m+1,n,2}^0+A_{m,n,2}^0=B_{m+1,n,1}^0-B_{m,n+1,1}^0+B_{m+1,n,2}^0-B_{m,n+1,2}^0,\\
\frac{B_{m,n,1}^0}{B_{m,n+1,1}^0}=\frac{A_{m-1,n+1,1}^0-A_{m,n,2}^0}{A_{m,n,1}^0-A_{m,n,2}^0},\qquad \frac{B_{m,n,2}^0}{B_{m+1,n,2}^0}=\frac{A_{m,n,1}^0-A_{m+1,n-1,2}^0}{A_{m,n,1}^0-A_{m,n,2}^0}.
\end{gather*}
Let us introduce the sequence of $2$-orthogonal polynomials $\{q_n^t\} $ by the correspondence
\begin{gather*}
q_{2n}^t(x)=P_{n,n}^t(x),\qquad q_{2n+1}^t(x)=P_{n+1,n}^t(x).
\end{gather*}
In a fashion similar to the previous section, we can also directly obtain the discrete spectral transformation of 2-orthogonal polynomials as follows
\begin{gather*}
(x-\lambda _t)q_{n}^{t+1}(x)=q_{n+1}^t(x)+X_{n}^tq_{m,n}^t(x),\qquad
q_n^t(x)=q_n^{t+1}(x)+Y_{n}^t q_{n-1}^{t+1}(x)+Z_{n}^tq_{n-2}^{t+1}(x),
\end{gather*}
with
\begin{gather}
X_{2n}^t=A_{n,n,1}^t,\qquad X_{2n+1}^t=A_{n+1,n,2}^t,\nonumber\\
Y_{2n}^t=B_{n,n,1}^t+B_{n,n,2}^t,\qquad Y_{2n+1}^t=B_{n+1,n,1}^t+B_{n+1,n,2}^t,\nonumber\\
Z_{2n}^t=B_{n,n,1}^t\big(A_{n-1,n,1}^t-A_{n,n-1,2}^t\big),\qquad
 Z_{2n+1}^t=B_{n+1,n,2}^t\big(B_{n+1,n-1}^{t}-A_{n,n,1}^{t}\big).\label{dmiura}
\end{gather}
From the compatibility condition of \eqref{d2lax} it follows that the discrete integrable system asso\-ciated with 2-orthogonal polynomials is given by
\begin{gather}
X_{n}^t+Y_{n+1}^t+\lambda_t =X_n^{t+1}+Y_n^{t+1}+\lambda_{t+1},\nonumber\\
X_n^tY_n^t+Z_{n+1}^t=X_{n-1}^{t+1}Y_n^{t+1}+Z_n^{t+1},\qquad X_n^tZ_n^t=X_{n-2}^{t+1}Z_n^{t+1}.\label{dktoda}
\end{gather}
It is easy to verify that the equation \eqref{dktoda} is the integrable discretization of the special case of Kostant--Toda equation~\eqref{ktoda} and the Miura transformation from~\eqref{d2gtoda} to~\eqref{dktoda} is explicitly given by~\eqref{dmiura}.

\subsection[The consistency approach: the stationary equations and the discrete-time dynamics]{The consistency approach: the stationary equations\\ and the discrete-time dynamics} \label{subsec:3.4}
To simplify formulas and statements we restrict ourselves here to the case where the multiple orthogonal polynomials are generated by two measures. Nevertheless, it can straightforwardly be generalized to the case $r>2$.

Once again, recall that multiple orthogonal polynomials are a generalization of orthogonal polynomials where the polynomials
are required to be simultaneously orthogonal with respect to two given measures \cite{MI,vanA2011}.
Now we consider a multi-index $ (n,m) \in{\mathbb Z}^2_+$ and suppose that~$\mu_1$,~$\mu_2$
are given positive measures on the real line. Then, the
type II multiple orthogonal polynomial is the monic polynomial $P_{n,m}(x) = x^{n+m} + \cdots$
of degree $n+m$ for which
\begin{gather*} 
 \int P_{n,m}(x) x^j d\mu_1(x) = 0, \qquad j=0,1,\ldots,n-1, \\
 \int P_{n,m}(x) x^j d\mu_2(x) = 0, \qquad j=0,1,\ldots,m-1.
\end{gather*}
As in the case of ordinary orthogonal polynomials, one can introduce the moments
\begin{gather*}
\mu_{j,i} = \int x^j d\mu_i(x),\qquad i=1,2,
\end{gather*}
and the determinant of the moment matrix
\begin{gather*} 
 \tau_{n,m} = \left| \begin{matrix}
 \mu_{0,1} & \cdots & \mu_{n-1,1}\\
 \mu_{1,1} & \cdots & \mu_{n,1} \\
	 \vdots & \cdots & \vdots \\
 \mu_{n+m-1,1} & \cdots & \mu_{2n+m-2,1} \end{matrix} \ 		
 \begin{matrix}
 \mu_{0,2} & \cdots & \mu_{m-1,2} \\
 \mu_{1,2} & \cdots & \mu_{m,2} \\
	 \vdots & \cdots & \vdots \\
 \mu_{n+m-1,2} & \cdots & \mu_{n+2m-2,2} 		
 \end{matrix} \right|.
\end{gather*}
Now we see that the type II multiple orthogonal polynomial can be written as
\begin{gather*} P_{n,m}(x) = \frac{1}{\tau_{n,m}}
 \left|\begin{matrix}
 \mu_{0,1} & \cdots & \mu_{n-1,1}\\
 \mu_{1,1} & \cdots & \mu_{n,1} \\
	 \vdots & \vdots & \vdots \\
 \mu_{n+m,1} & \cdots & \mu_{2n+m-1,1}\end{matrix}	 \	
 \begin{matrix}
 \mu_{0,2} & \cdots & \mu_{m-1,2} \\
 \mu_{1,2} & \cdots & \mu_{m,2} \\
	 \vdots & \cdots & \vdots \\
 \mu_{n+m,2} & \cdots & \mu_{n+2m-1,2}		
 \end{matrix} \
 \begin{matrix} 1 \\ x \\ \vdots \\ x^{n+m} \end{matrix} \right|
\end{gather*}
provided that $\tau_{n,m}$ is nonvanishing. In the latter case the index $(n,m)$ is normal.
We assume that all multi-indices are normal.

In the case of multiple orthogonal polynomials, the three-term recurrence relations are replaced with the following relation for the nearest neighbors
\begin{gather}
 P_{n+1,m}(x) = (x-b_{n,m,1})P_{n,m}(x) - a_{n,m,1} P_{n-1,m}(x) - a_{n,m,2} P_{n,m-1}(x),\nonumber \\
 P_{n,m+1}(x) = (x-b_{n,m,2})P_{n,m}(x) - a_{n,m,1} P_{n-1,m}(x) - a_{n,m,2} P_{n,m-1}(x),\label{eq:2.6}
\end{gather}
with $a_{0,m,1}=0$ and $a_{n,0,2}=0$ for all $n,m \geq 0$.

Unlike the case of ordinary orthogonal polynomials, the coef\/f\/icients of the recurrence rela\-tions~\eqref{eq:2.6} are solutions of a discrete integrable system even without introducing the discrete time evolution. Here we follow the concept of discrete integrability given in~\cite{BS2002}.

\begin{Proposition}[\cite{ADvanA2015,ADvanA2014}] Let us consider the following vector-valued wave function
\begin{gather*}
\Psi_{n,m}(x)=\big(P_{n,m}(x),P_{n-1,m}(x),P_{n,m-1}(x)\big)^{\top}.
\end{gather*}
Then the corresponding transition matrices are
\begin{gather*} 
 L_{n,m} = \begin{pmatrix}
 z-b_{n,m,1} & -a_{n,m,1} & -a_{n,m,2} \\
 1 & 0 & 0 \\
 1 & 0 & b_{n,m-1,2}-b_{n,m-1,1}
 \end{pmatrix}
\end{gather*}
and
\begin{gather*} 
 M_{n,m} = \begin{pmatrix}
 z-b_{n,m,2} & -a_{n,m,1} & -a_{n,m,2} \\
 1 & b_{n-1,m,1}-b_{n-1,m,2} & 0 \\
 1 & 0 & 0
 \end{pmatrix},
\end{gather*}
and they give the non-trivial zero curvature condition
\begin{gather}\label{MOPDIS}
0 = L_{n,m+1} M_{n,m} - M_{n+1,m} L_{n,m}.
\end{gather}
\end{Proposition}
Indeed, it follows from \eqref{eq:2.6} that
\begin{gather} \label{eq:2.9}
 \Psi_{n+1,m} =L_{n,m} \Psi_{n,m}, \qquad \Psi_{n,m+1} = M_{n,m} \Psi_{n,m}.
\end{gather}
It is now clear that the consistency of \eqref{eq:2.9} gives \eqref{MOPDIS}, which is in fact a discrete integrable system \cite{ADvanA2015, ADvanA2014}. Namely, in \cite{ADvanA2014} and \cite{vanA2011} it is shown that the discrete zero curvature condi\-tion~\eqref{MOPDIS} is equivalent to the nonlinear system of dif\/ference equations~\eqref{diffcoeff} for the coef\/f\/icients of the recurrence relations~\eqref{eq:2.6}. Furthermore we have the following formulas for the recurrence coef\/f\/icients~\cite{vanA2011}
\begin{gather}
a_{n,m,1} = \frac{\tau_{n+1,m} \tau_{n-1,m}}{\big(\tau_{n,m}\big)^2}, \qquad
a_{n,m,2} = \frac{\tau_{n,m+1} \tau_{n,m-1}}{\big(\tau_{n,m}\big)^2}, \nonumber\\
b_{n,m,2}-b_{n,m,1} = \frac{\tau_{n,m} \tau_{n+1,m+1}}{\tau_{n+1,m} \tau_{n,m+1}}.\label{CoeffDet}
\end{gather}
Nevertheless, they do not determine the coef\/f\/icients of the recurrence relations~\eqref{eq:2.6} from the moments of the given measures. Still, this obstacle can easily be overcome.
\begin{Proposition} We have that
\begin{gather}
b_{n,m+1,1} = b_{n,0,1} +\sum_{i=1}^m \frac{(a_{n+1,i,1}+a_{n+1,i,2}) -
(a_{n,i+1,1}+a_{n,i+1,2})}{(b_{n,i,1}-b_{n,i,2})}, \nonumber\\
b_{n+1,m,2} = b_{n,0,2} +\sum_{i=1}^n \frac{(a_{i+1,m,1}+a_{i+1,m,2}) -
(a_{i,m+1,1}+a_{i,m+1,2})}{(b_{i,m,1}-b_{i,m,2})},\label{CoeffDetcd}
\end{gather}
where the right hand sides can be obtained from the moments by~\eqref{CoeffDet} and~\eqref{OPRecFromDet}.
\end{Proposition}
\begin{proof}
The relations \eqref{CoeffDetcd} are obtained from the discrete zero curvature condition~\eqref{MOPDIS} (see also~\eqref{diffcoeff}) by summation of the corresponding relations for consecutive indices.
\end{proof}

Now we re-derive the dm-Toda equations \eqref{gtoda} that we already obtained in Subsection~\ref{subsec:3.4}. However, in this case we follow the consistency approach from~\cite{BS2002} and~\cite{SNderK2011}. In particular, we get the Lax pair here by using a method that is the adaptation of the one from~\cite{SNderK2011} (see Proposi\-tion~\ref{LaxPairToda}).
Since we only consider the case of two measures, we need to consider the family of two measures $x^t d\mu_1(x)$ and $x^t d\mu_2(x)$, where $t\in{\mathbb Z}_+$ is the discrete time. In other words, we have two sequences of moments $\big\{s_j^{(1)}\big\}_{j=0}^{\infty}$
and $\big\{s_j^{(2)}\big\}_{j=0}^{\infty}$ and we consider their truncations
\begin{gather*}
\big\{s_{j+t}^{(1)}\big\}_{j=0}^{\infty},\qquad \big\{s_{j+t}^{(2)}\big\}_{j=0}^{\infty},\qquad t\in{\mathbb Z}_+,
\end{gather*}
which are actually given by the measures $x^t d\mu_1(x)$ and $x^t d\mu_2(x)$ Clearly, these sequences of moments generate a family of multiple orthogonal polynomials, which, as we have already seen, have the following form
\begin{gather*}
P_{n,m}^{t}(x) = \frac{1}{\tau_{n,m}^{t}}
 \left|\begin{matrix}
 \mu_{t,1} & \cdots & \mu_{t+n-1,1}\\
 \mu_{t+1,1} & \cdots & \mu_{t+n,1} \\
	 \vdots & \cdots & \vdots \\
 \mu_{t+n+m,1} & \cdots & \mu_{t+2n+m-1,1} \end{matrix}	\	
 \begin{matrix}
 \mu_{t,2} & \cdots & \mu_{t+m-1,2} \\
 \mu_{t+1,2} & \cdots & \mu_{t+m,2} \\
	 \vdots & \cdots & \vdots \\
 \mu_{t+n+m,2} & \cdots & \mu_{t+n+2m-1,2} 		
 \end{matrix} \
 \begin{matrix} 1 \\ x \\ \vdots \\ x^{n+m} \end{matrix} \right|,
\end{gather*}
with
\begin{gather*} 
 \tau_{n,m}^{t} = \left| \begin{matrix}
 \mu_{t,1} & \cdots & \mu_{t+n-1,1}\\
 \mu_{t+1,1} & \cdots & \mu_{t+n,1} \\
	 \vdots & \cdots & \vdots \\
 \mu_{t+n+m-1,1} & \cdots & \mu_{t+2n+m-2,1} \end{matrix}	\	
 \begin{matrix}
 \mu_{t,2} & \cdots & \mu_{t+m-1,2} \\
 \mu_{t+1,2} & \cdots & \mu_{t+m,2} \\
	 \vdots & \cdots & \vdots \\
 \mu_{t+n+m-1,2} & \cdots & \mu_{t+n+2m-2,2} 		
 \end{matrix} \right|.
\end{gather*}
As a matter of fact, we obtained an analogue of the Christof\/fel transformation in the case of multiple orthogonal polynomials in
Proposition~\ref{CTforMPh}. Nevertheless, let's do it again but this time we apply the following two dif\/ferent forms of the Sylvester identity
\begin{gather*}
|A||A_{1,n+m+1;n+m,n+m+1}|=|A_{1;n+m}||A_{n+m+1;n+m+1}|-|A_{n+m+1;n+m}||A_{1;n+m+1}|,\\
|A||A_{1,n+m+1;n,n+m+1}|=|A_{1;n}||A_{n+m+1;n+m+1}|-|A_{n+m+1;n}||A_{1;n+m+1}|,
\end{gather*}
to the determinant $\tau_{n,m}^{t}P_{n,m}^{t}(x)$. Evidently, this leads to the relations
\begin{gather}
P_{n,m}^{t}(x)=xP_{n,m-1}^{t+1}(x)-A_{n,m-1,2}^{t}P_{n,m-1}^{t}(x),\nonumber\\
P_{n,m}^{t}(x)=xP_{n-1,m}^{t+1}(x)-A_{n-1,m,1}^{t}P_{n-1,m}^{t}(x),\label{GPT}
\end{gather}
where the coef\/f\/icients are def\/ined by the formulas
\begin{gather}\label{AcoefForLT}
A_{n,m-1,2}^{t}=\frac{\tau_{n,m-1}^{t}\tau_{n,m}^{t+1}}{\tau_{n,m}^{t}\tau_{n,m-1}^{t+1}},\qquad
A_{n-1,m,1}^{t}=\frac{\tau_{n-1,m}^{t}\tau_{n,m}^{t+1}}{\tau_{n,m}^{t}\tau_{n-1,m}^{t+1}}.
\end{gather}
Now, based on the relations \eqref{eq:2.6} and \eqref{GPT}, we can extend Proposition~\ref{LaxPairToda} to the context of multiple orthogonal polynomials. Thus, we are in the position to complete the associated discrete integrable system~\eqref{MOPDIS} on ${\mathbb Z}_+^2$ to a discrete integrable system on ${\mathbb Z}_+^3$. To this end, we f\/irst obtain the following relations
\begin{gather}
 P_{n-1,m}^{t+1}(x)=\frac{1}{x}P_{n,m}^{t}(x) +\frac{A_{n-1,m,1}^{t}}{x}P_{n-1,m}^{t}(x),\nonumber\\
 P_{n,m-1}^{t+1}(x)=\frac{1}{x}P_{n,m}^{t}(x) +\frac{A_{n,m-1,2}^{t}}{x}P_{n,m-1}^{t}(x),\nonumber\\
 P_{n,m}^{t+1}(x)=\left(1-\frac{b_{n,m,2}^{t}-A_{n,m,2}^{t}}{x}\right)P_{n,m}^{t} -\frac{a_{n,m,1}^{t}}{x}P_{n-1,m}^{t}-\frac{a_{n,m,2}^{t}}{x}P_{n,m-1}^{t},\label{CT1}
\end{gather}
by manipulations with \eqref{GPT} and \eqref{eq:2.6}.

Now we see that we have a lot of options to travel over~${\mathbb Z}_+^3$ using the above-given relations. It is obvious that we don't have to use all of them to do that. However, applying dif\/ferent formulas when moving along the same path leads to consistency relations and the following statement contains all of them.

\begin{Theorem}
Let us consider the vector-valued wave function
\begin{gather*}
\Psi_{n,m,t}(x)=\big(P_{n,m}^{t}(x),P_{n-1,m}^{t}(x),P_{n,m-1}^{t}(x)\big)^{\top}.
\end{gather*}
Then three families of matrices given by the formulas
\begin{gather*} 
 L_{n,m,t} = \begin{pmatrix}
 x-b^{t}_{n,m,1} & -{a^{t}_{n,m,1}} & -{a^{t}_{n,m,2}} \\
 1 & 0 & 0 \\
 1 & 0 & {b^{t}_{n,m-1,2}}-b^{t}_{n,m-1,1}
 \end{pmatrix},
\\
 M_{n,m,t} = \begin{pmatrix}
 x-{b^{t}_{n,m,2}} & -{a^{t}_{n,m,1}} & -{a^{t}_{n,m,2}} \\
 1 & b^{t}_{n-1,m,1}-{b^{t}_{n-1,m,2}} & 0 \\
 1 & 0 & 0
 \end{pmatrix},
\end{gather*}
and
\begin{gather*} 
 N_{n,m,t} = \begin{pmatrix}
 1-\frac{{b^{t}_{n,m,2}}-A_{n,m,2}^{t}}{x} & -\frac{{a^{t}_{n,m,1}}}{x} & -\frac{{a^{t}_{n,m,2}}}{x} \\
 \frac{1}{x} & \frac{A_{n-1,m,1}^{t}}{x} & 0 \\
 \frac{1}{x} & 0 & \frac{A_{n,m-1,2}^{t}}{x}
 \end{pmatrix}
\end{gather*}
are the transition matrices for $\Psi_{n,m,k}$ and they satisfy the following relations
\begin{gather}
0 = L_{n,m+1,t} M_{n,m,t} - M_{n+1,m,t} L_{n,m,t},\nonumber\\
0= M_{n,m,t+1} N_{n,m,t} - N_{n,m+1,t} M_{n,m,t}, \nonumber\\
0= L_{n,m,t+1} N_{n,m,t} - N_{n+1,m,t} L_{n,m,t},\label{MOPTodaDis}
\end{gather}
which give the discrete zero curvature condition.
\end{Theorem}
\begin{proof}
To begin with, let us notice that the relations \eqref{eq:2.6}, and \eqref{CT1} can be used to get the following vector equalities
\begin{gather} \label{eq:2.9_3d}
 \Psi_{n+1,m,t} =L_{n,m,t} \Psi_{n,m,t}, \!\qquad \Psi_{n,m+1,t} = M_{n,m,t} \Psi_{n,m,t},\!\qquad \Psi_{n,m,t+1} =N_{n,m,t} \Psi_{n,m,t}.\!\!\!
\end{gather}
The latter system means that the matrices $L_{n,m,t}$, $M_{n,m,t}$, and $N_{n,m,t}$ are transition matrices for the wave function~$\Psi_{n,m,t}$.
Next, one can easily see that the consistency of~\eqref{eq:2.9_3d} leads to the following relations
\begin{gather}
L_{n,m+1,t} M_{n,m,t}\Psi_{n,m,t}=M_{n+1,m,t} L_{n,m,t}\Psi_{n,m,t},\nonumber\\
M_{n,m,t+1} N_{n,m,t}\Psi_{n,m,t}=N_{n,m+1,t} M_{n,m,t}\Psi_{n,m,t}, \nonumber\\
L_{n,m,t+1} N_{n,m,t}\Psi_{n,m,t}= N_{n+1,m,t} L_{n,m,t}\Psi_{n,m,t}.\label{help3Dproof}
\end{gather}
Now to get \eqref{MOPTodaDis} it remains to observe that the polynomials in the vector $\Psi_{n,m,t}$ are linearly independent whenever $ (n,m)$ is a normal index at the moment $t$. Indeed, if there are numbers $\alpha_1$, $\alpha_2$, and $\alpha_3$ such that
\begin{gather*}
\alpha_1 P_{n,m}^{t} + \alpha_2 P_{n-1,m}^{t} +\alpha_3 P_{n,m-1}^{t}= 0,
\end{gather*}
then it follows by comparing the leading coef\/f\/icients that $\alpha_1=0$. Next, if one multiplies that relation by $x^{n-1}$ and integrates the resulting relation with respect to $x^k d\mu_1(x)$ then one gets
\begin{gather*}
\alpha_2 \int x^{n-1} P_{n-1,m}^{t}(x) x^t d\mu_1(x) = 0
\end{gather*}
since $P_{n,m-1}^{t}$ is orthogonal to $x^{n-1}$ with respect to $x^k d\mu_1(x)$ by def\/inition. Hence, $\alpha_2=0$ because of the normality of the index $(n,m,t)$. In other words, normality means that the determi\-nant~$\tau_{n,m}^{t}$ is nonvanishing. On the other hand, it is not so hard to see that
\begin{gather*}
\int x^{n-1} P_{n-1,m}^{t}(x) x^t d\mu_1(x) = \epsilon \tau_{n,m}^{t},
\end{gather*}
where $\epsilon=\pm 1$. Analogously, one can get that $\alpha_3=0$.
The relations \eqref{help3Dproof} reduce to \eqref{MOPTodaDis}, which is a discrete integrable system on ${\mathbb Z}_+^3$.
\end{proof}

\begin{Remark}
The statement of the above theorem shows that the system of dif\/ference equations obtained from \eqref{MOPTodaDis} is an integrable system in the sense of \cite{BS2002}. However, one might still wonder about the relation between~\eqref{MOPTodaDis} and~\eqref{Int_DToda}. Basically, \eqref{MOPTodaDis} is a~representation of~\eqref{Int_DToda} by means of a certain Lax pair.
\end{Remark}
Indeed, this is the case and we are going to show how one can get representative equations of~\eqref{Int_DToda} from the Lax pair representation~\eqref{MOPTodaDis}. To this end, let us consider entry~$(1,3)$ of the relation
\begin{gather*}
0=M_{n,m,t+1} N_{n,m,t} - N_{n,m+1,t} M_{n,m,t}.
\end{gather*}
More precisely, the entry in question gives
\begin{gather*}
0=-\big(x-b^{t+1}_{n,m,2}\big)\frac{{a^{t}_{n,m,2}}}{x}-{a_{n,m,2}}^{t+1} \frac{A_{n,m-1,2}^{t}}{x} + \left(1-\frac{b^{t}_{n,m+1,2}-A_{n,m+1,2}^{t}}{x}\right){a^{t}_{n,m,2}},
\end{gather*}
which is equivalent to
\begin{gather*}
b^{t+1}_{n,m,2}-b^{t}_{n,m+1,2}+A_{n,m+1,2}^{t}=\frac{a^{t+1}_{n,m,2}}{a^{t}_{n,m,2}}A_{n,m-1,2}^{t}.
\end{gather*}
Now, taking into account the f\/irst relation in \eqref{AcoefForLT} and the second one in \eqref{CoeffDet}, one can see that
\begin{gather*}
\frac{a^{t+1}_{n,m,2}}{a^{t}_{n,m,2}}A_{n,m-1,2}^{t}=A_{n,m,2}^{t},
\end{gather*}
and, therefore, we arrive at
\begin{gather}\label{Entry31}
b^{t+1}_{n,m,2}-b^{t}_{n,m+1,2}=A_{n,m,2}^{t}-A_{n,m+1,2}^{t}.
\end{gather}

Next, we know from \eqref{tcoeff} that
\begin{gather*}
b^{t}_{n,m,2}=A_{n,m,2}^t+\sum_{k=1}^2B_{n,m,k}^t.
\end{gather*}
As a consequence, it follows from \eqref{Entry31} that
\begin{gather*}
A_{n,m,2}^{t+1}+\sum_{k=1}^2B_{n,m,k}^{t+1}=A_{n,m,2}^t+\sum_{k=1}^2B_{n,m+1,k}^t,
\end{gather*}
which is one of the equations from~\eqref{Int_DToda}. To get another equation in~\eqref{Int_DToda}, let us take a look at entry $(1,2)$ of the relation
\begin{gather*}
0=M_{n,m,t+1} N_{n,m,t} - N_{n,m+1,t} M_{n,m,t}.
\end{gather*}
The entry in question gives
\begin{gather*}
0=-\big(x-b^{t+1}_{n,m,2}\big)\frac{a^{t}_{n,m,1}}{x}-a^{t+1}_{n,m,1} \frac{A_{n-1,m,1}^{t}}{x}+\left(1-\frac{b^{t}_{n,m+1,2}-A_{n,m+1,2}^{t}}{x}\right)
a^{t}_{n,m,1}\\
\hphantom{0=}{}
- \frac{a^{t}_{n,m+1,1}}{x}\big(b^{t}_{n-1,m,1}-b^{t}_{n-1,m,2}\big),
\end{gather*}
which is equivalent to
\begin{gather*}
b^{t+1}_{n,m,2}-b^{t}_{n,m+1,2}+A_{n,m+1,2}^{t}=\frac{a^{t+1}_{n,m,1}}{a^{t}_{n,m,1}}A_{n-1,m,1}^{t}-
\frac{a^{t}_{n,m+1,1}}{a^{t}_{n,m,1}}\big(b^{t}_{n-1,m,1}-b^{t}_{n-1,m,2}\big).
\end{gather*}
Now, the second relation in \eqref{diffcoeff} gives
\begin{gather*}
\frac{a^{t}_{n,m+1,1}}{a^{t}_{n,m,1}}\big(b^{t}_{n-1,m,1}-b^{t}_{n-1,m,2}\big)=b^{t}_{n,m,1}-b^{t}_{n,m,2}.
\end{gather*}
Hence, we arrive at
\begin{gather}\label{Entry21}
b^{t+1}_{n,m,2}-b^{t}_{n,m+1,2}+A_{n,m+1,2}^{t}=\frac{a^{t+1}_{n,m,1}}{a^{t}_{n,m,1}}A_{n-1,m,1}^{t}+
b^{t}_{n,m,1}-b^{t}_{n,m,2}.
\end{gather}
Since by \eqref{tcoeff} we have
\begin{gather*}
b^{t}_{n,m,1}=A_{n,m,1}^t+\sum_{k=1}^2B_{n,m,k}^t, \qquad
b^{t}_{n,m,2}=A_{n,m,2}^t+\sum_{k=1}^2B_{n,m,k}^t,
\end{gather*}
due to \eqref{Entry31} the relation \eqref{Entry21} reduces to
\begin{gather*}
A_{n,m,2}^{t}-A_{n,m+1,2}^{t}+A_{n,m+1,2}^{t}=\frac{a^{t+1}_{n,m,1}}{a^{t}_{n,m,1}}A_{n-1,m,1}^{t}-A_{n,m,1}^{t}+A_{n,m,2}^{t}
\end{gather*}
or, equivalently,
\begin{gather}\label{TheSecond}
{a^{t+1}_{n,m,1}}=\frac{A_{n,m,1}^{t}}{A_{n-1,m,1}^{t}}a^{t}_{n,m,1}.
\end{gather}
Next, according to \eqref{tcoeff} we have
\begin{gather*}
a^{t}_{n,m,1}=A_{n-1,m,1}^{t}B_{n,m,1}^{t}.
\end{gather*}
Hence \eqref{TheSecond} reduces to
\begin{gather*}
A_{n-1,m,1}^{t+1}B_{n,m,1}^{t+1}=A_{n,m,1}^{t}B_{n,m,1}^{t},
\end{gather*}
which is clearly another one from \eqref{Int_DToda}. In other words, we have reached the following conclusion.
\begin{Proposition} The second and third relations in \eqref{MOPTodaDis} are one of the Lax representations of \eqref{Int_DToda}.
\end{Proposition}

\subsection{The explicit solution of the dm-Toda equation}~\label{subsec:3.5}
To illustrate our approach and provide the reader with an explicit example when the scheme can be applied, let us recall that multiple Laguerre polynomials of the second kind are given by the orthogonality relations
\begin{gather*}
\int_0^\infty x^k L_{n,m}^\alpha(x) x^{\alpha} e^{-c_j x} dx = 0 , \qquad k = 0, 1, \ldots, n_j-1,
\end{gather*}
for $j=1,2$, where $\alpha > -1$, $c_1$, $c_2 > 0$ and $c_1 \neq c_2$. Evidently, putting $t=\alpha$ and denoting $P_{n,m}^{t}(x)=L_{n,m}^t(x)$ we get the polynomials with the discrete-time dynamics
\begin{gather*}
d\mu_1(t,x)=x^te^{-c_1 x} dx, \qquad d\mu_1(t,x)=x^te^{-c_2 x} dx,
\end{gather*}
and the corresponding coef\/f\/icients
\begin{gather*}
a^{t}_{n,m,1}=\frac{n+m+t}{c_1^2}m,\qquad a^{t}_{n,m,2}=\frac{n+m+t}{c_2^2}n,\\
b^{t}_{n,m,1}=\frac{n+m+t}{c_1}+\frac{n}{c_1}+\frac{m}{c_2},\qquad b^{t}_{n,m,2}=\frac{n+m+t}{c_2}+\frac{n}{c_1}+\frac{m}{c_2}
\end{gather*}
are a solution of \eqref{MOPTodaDis}. The multiple Laguerre polynomials $L_{n_1,n_2}^t$ of the second kind can be obtained using the Rodrigues formula
\begin{gather*}
 (-1)^{n_1+n_2} \left( \prod_{j=1}^2 c_j^{n_j} \right) x^t L_{n_1,n_2}^t(x)
 = \prod_{j=1}^2 \left( e^{c_jx} \frac{d^{n_j}}{dx^{n_j}} e^{-c_jx} \right) x^{n_1+n_2+t} ,
\end{gather*}
where the dif\/ferential operators in the product can be taken in any order \cite{ABVA}.

\subsection*{Acknowledgements} A.I.~Aptekarev was supported by grant RScF-14-21-00025. M.~Derevyagin thanks the hospitality of Department of Mathematics of KU Leuven, where his part of the research was initiated while he was a postdoc there. M.~Derevyagin and W.~Van Assche gratefully acknowledge the support of FWO Flanders project G.0934.13, KU Leuven research grant OT/12/073 and the Belgian Interuniversity Attraction Poles
programme P07/18. H.~Miki was supported by JSPS KAKENHI Grant Number 15K17561. Also, M.~Derevyagin and H.~Miki are grateful to S.~Tsujimoto, L.~Vinet, A.~Zhedanov for valuable discussions and comments. Finally, all the authors thank the anonymous referees for their careful reading of the manuscript and for their remarks that improved the presentation of the paper.

\pdfbookmark[1]{References}{ref}
\LastPageEnding

\end{document}